\documentclass[journal]{IEEEtran}
\usepackage{amsmath,amsfonts,amssymb,dsfont}

\usepackage{hyperref}
\usepackage[nocompress]{cite}

\usepackage{algorithm}
\usepackage{algpseudocode}

\usepackage{booktabs}
\usepackage{multirow}
\usepackage{tabularx}

\usepackage[caption=false,font=footnotesize]{subfig}
\usepackage{graphicx}
\usepackage{tikz, pgfplots}
\usetikzlibrary{shapes, arrows, arrows.meta, fit, positioning, calc, 3d}
\usepgfplotslibrary{fillbetween}
\pgfplotsset{compat=1.18} 
\tikzset{>=latex}

\DeclareMathOperator*{\argmax}{arg\,max}
\DeclareMathOperator*{\argmin}{arg\,min}
\providecommand*{\T}[1]{\mathrm{#1}}  

\newcommand{\set}[1]{\mathcal{#1}}
\newtheorem{theorem}{Theorem}
\newtheorem{lemma}{Lemma}
\newtheorem{definition}{Definition}
\newtheorem{proposition}{Proposition}
\newtheorem{assumption}{Assumption}
\newtheorem{corollary}{Corollary}

\allowdisplaybreaks

\begin{document}
\bstctlcite{IEEEexample:BSTcontrol}

\title{On the Role of Age and Semantics of Information in Remote Estimation of Markov Sources}
\author{Jiping~Luo,~Nikolaos~Pappas,~\IEEEmembership{Senior~Member,~IEEE}
\thanks{This work has been supported in part by ELLIIT, the Graduate School in Computer Science (CUGS), the European Union (ROBUST-6G, 101139068, and 6G-LEADER, 101192080), and the European Union's Horizon Europe research and innovation programme under the Marie Skłodowska-Curie Grant Agreement No 101131481 (SOVEREIGN). \textit{(Corresponding author: Jiping~Luo.)}}
\thanks{The authors are with the Department of Computer and Information Science, Link\"oping University, 58183 Link\"oping, Sweden (e-mail: jiping.luo@liu.se; nikolaos.pappas@liu.se).}
}

\maketitle
\begin{abstract}
    This paper studies semantics-aware remote estimation of Markov sources. We leverage two complementary information attributes: the urgency of lasting impact, which quantifies the \emph{significance} of consecutive estimation error at the transmitter, and the age of information (AoI), which captures the \emph{predictability} of outdated information at the receiver. The objective is to minimize the long-run average lasting impact subject to a transmission frequency constraint. The problem is formulated as a constrained Markov decision process (CMDP) with potentially unbounded costs. We show the existence of an optimal \emph{simple mixture policy}, which randomizes between two neighboring \emph{switching policies} at a common regeneration state. A closed-form expression for the optimal mixture coefficient is derived. Each switching policy triggers transmission only when the error holding time exceeds a threshold that depends on both the instantaneous estimation error and the AoI. We further derive sufficient conditions under which the thresholds are independent of the instantaneous error and the AoI. Finally, we propose a structure-aware algorithm, Insec-SPI, that computes the optimal policy with reduced computation overhead. Numerical results demonstrate that incorporating both the age and semantics of information significantly improves estimation performance compared to using either attribute alone.
\end{abstract}
\begin{IEEEkeywords}
Age and semantics of information, constrained Markov decision process, maximum a posteriori estimator.
\end{IEEEkeywords}

\section{Introduction}
\IEEEPARstart{A}{chieving} desired estimation quality under constrained communication resources is crucial in many networked control systems (NCSs)~\cite{Brockett-TAC-1997, Schenato-ProcIEEE-2007, WNCSSurvey-2018}. In classical remote estimation, estimation quality is equated with \emph{accuracy}, and the primary objective is to minimize distortion between the source and the reconstructed signal~\cite {Lipsa-TAC-2011-scalar-source, lingshi2013TAC, Aditya-TAC-2017, WuShuang-TCNS-2020}. In emerging cyber-physical systems, however, the timeliness and contextual relevance of information often outweigh mere signal fidelity. This motivates the semantics-aware estimation, where the focus shifts from minimizing distortion alone to ensuring that the information conveyed is \emph{fresh}, \emph{significant}, and aligned with the application context and goals~\cite{nikos2021CM, luo2025information}.

This paper studies the remote estimation of a finite state Markov chain $\{X_t\}_{t\geq 1}$, as depicted in Figure~\ref{fig:system_model}. A sensor, characterized by a transmission policy $\pi$, decides when to send source updates to a remote receiver, which operates under a maximum a posteriori (MAP) estimator $g$. The wireless channel is unreliable and subject to a transmission frequency constraint. Consequently, there might be a discrepancy between the source $X_t$ and the reconstructed signal $\hat{X}_t$. 

\begin{figure}[t]
    \centering
    \scalebox{0.83}{\begin{tikzpicture}[scale=1.0]
    \node [draw, rounded corners=1.5pt, rectangle, minimum width=1.5cm, minimum height=0.8cm] (source) {Source};
    \node [draw, rounded corners=1.5pt, rectangle, right=0.9cm of source, minimum width=1.5cm, minimum height=0.8cm] (sensor) {Sensor};
    \node [draw, ellipse, minimum width=0.8cm, minimum height=0.8cm, right=0.8cm of sensor] (channel) {Channel};
    \node [draw, rounded corners=1.5pt, rectangle, right=1.cm of channel, minimum width=1.5cm, minimum height=0.8cm] (receiver) {Receiver};
    \node [right=0.8cm of receiver] (end) {};
    
    \draw[->] (source) -- node[above] {$X_t$} (sensor);
    \draw[->] (sensor) -- node[above] {} (channel);
    \draw[->] (channel) -- node[above] {$Y_t$} coordinate[midway] (mid-cd) (receiver);
    \draw[->] (receiver) -- node[above] {$\hat{X}_t$} (end);
    
    \draw[->, dashed] (mid-cd) -- ++(0,-0.8) -| 
    node[pos=0.25, below] {} 
    ($(sensor.south west)!0.5!(sensor.south east)$);
    \node[above=0.cm of sensor] {$\pi_t$};
    \node[above=0.cm of receiver] {$g_t$};
\end{tikzpicture}}
    \caption{Remote state estimation of a Markov source.}
    \label{fig:system_model}
\end{figure}

We leverage the semantic attribute of the \emph{urgency of lasting impact} to quantify the information value so as to filter out less important measurements. Specifically, the severity of an error depends on both its contextual significance and its duration (see, e.g.,~\cite{assaad2020TON, luo2026exploiting, luo2025cost}). This attribute is critical in many applications. For instance, in autonomous driving~\cite{jiping2023TITS}, the cost of misdetecting a nearby obstacle may grow exponentially as the error persists. Let $\Delta_t$ denote the number of consecutive time slots during which the system has remained in an erroneous state $(X_t, \hat{X}_t)$. The lasting impact is measured by a collection of context-aware nonlinear functions $\rho_{X_t, \hat{X}_t}(\Delta_t)$, where $\rho_{i,j}$ are general non-decreasing functions that represent escalating penalties for prolonged erroneous states.

We aim to find a communication policy that minimizes the long-run average lasting impact of consecutive estimation errors, subject to a hard constraint on the transmission frequency. Our main contributions are as follows:

(1) We study the complementary roles of information age and information semantics in remote estimation systems. The MAP estimator is characterized by both the age and the content of the most recently received measurement. By leveraging the urgency of lasting impact and the \emph{age} of information (AoI), we quantify both the significance of estimation errors at the transmitter and the usefulness of aged information at the receiver, thereby avoiding unnecessary transmissions for predictable or low-impact errors.

(2) We characterize the structure of the optimal policy. It admits a simple mixture structure that randomizes between two neighboring switching policies with a fixed probability. Each switching policy triggers transmission only when the error duration exceeds a threshold that depends on both the instantaneous estimation error and the AoI. We also provide sufficient conditions under which the thresholds are independent of the estimation error and the AoI. Moreover, the two neighboring policies differ in exactly one threshold, and we derive a closed-form expression for the mixture coefficient.

(3) We leverage these structural results to develop an efficient algorithm, Insec-SPI, which solves the constrained optimization problem in significantly fewer iterations than existing unstructured algorithms. Numerical results validate our theoretical findings and demonstrate that incorporating both the age and semantics of information yields significant improvements in estimation performance.

The rest of the paper is organized as follows. Section~\ref{sec:related-works} reviews related work. Section~\ref{sec:model} introduces the system model and problem formulation. Section~\ref{sec:main-results} presents the main results on the optimal solution. Section~\ref{sec:computation} presents the structure-aware Insec-SPI algorithm. The numerical results and the conclusion are provided in Section~\ref{sec:simulations} and Section~\ref{sec:conclusion}.

\section{Related Work}\label{sec:related-works}
In many real-time status update systems, information holds the greatest value when it is fresh. The notion of AoI measures the freshness of information that a receiver has about a source~\cite{Roy2012INFOCOM, kosta2017age, RoyYates2021JSAC}. Recent studies have demonstrated the role of AoI in remote estimation systems~\cite{WuShuang-TCNS-2020, Telatar2021ITW, Melih2021TON, Jayanth2023WiOpt, vishrant2025TMC, jiping2025wearing}. For example, in linear Gaussian source estimation, the error covariance is a monotonic function of AoI~\cite{WuShuang-TCNS-2020, jiping2025wearing}. However, AoI ignores the application context and has been considered an inefficient metric for monitoring general sources such as Wiener processes~\cite{SunYin-TIT-2020}, Ornstein-Uhlenbeck processes~\cite{SunYin-TON-2021}, and Markov chains~\cite{george2019GCWkshps, assaad2020TON, luo2025semantic}. As will be shown in this paper, however, this holds only for zero-order hold (ZOH) estimators that retain the latest received measurement as their estimate. For the MAP estimator, AoI remains relevant because it quantifies the usefulness of outdated information at the receiver~\cite{ayan2025age, cosandal2025joint}.

Beyond accuracy and freshness, semantic communication transmits the most significant piece of information, thereby reducing the amount of less important data transmitted in the network~\cite{nikos2021CM, luo2025information}. A fundamental problem is to define key semantic attributes that guide the generation, transmission, and reconstruction of information to achieve specific goals or tasks. Some related works are discussed below:

Information quality depends on both its age and content. Metrics such as content-aware AoI~\cite{george2019GCWkshps, ErfanTGCN24}, uncertainty of information~\cite{gongpu2022UoI, chen2024index}, state estimation entropy~\cite{Andrea2023JSAIT}, and age of channel state information~\cite{ayan2025age, costa2015age} reveal that information quality evolves with AoI at different rates depending on the content. Content-aware distortion metrics (see, e.g.,~\cite{nikos2021ICAS, mehrdad2024TCOM, Mehrdad-JCN-2023, luo2024goal, luo2025semantic, Zakeri-Asilomar-2023, Niu2020TWC, SunYin2023MILCOM}) account for the fact that the cost of estimation error depends not only on the physical discrepancy but also on the contextual relevance and potential control risks to system performance. Structural results for resource-constrained multi-source systems are derived in~\cite{luo2025semantic} and~\cite{luo2024goal}, offering insights into multimodal scenarios.

The persistence cost of consecutive errors may have a significant impact on many NCSs, yet it is rarely considered in classical remote estimation literature. The age of incorrect information (AoII)~\cite{assaad2020TON, ChenYutao-TON-AoII-2024, Tony2025TON, Nail-ISIT-2024, cosandal2025joint, Andrea2024distributed} and the cost of memory error~\cite{mehrdad2024TCOM} penalize the system through cost functions that scale with the time elapsed since the system was last synced. However, most existing studies adopt a context-agnostic form, which may be insufficient for applications where different estimation errors have fundamentally distinct implications and should be treated separately. The age of missed alarm (AoMA) and the age of false alarm (AoFA)~\cite{luo2024minimizing, luo2026exploiting} are context-aware metrics that quantify the lasting impacts of missed and false alarms of a binary Markov chain. Similar ideas have been applied to NCS stability~\cite{kriouile2025semantics}. Moreover, the severity of an estimation error depends on both its duration and context. The significance-aware age of consecutive error (AoCE)~\cite{luo2025cost} thus measures the urgency of lasting impact for general Markov chains through a set of context-aware nonlinear functions. Analytical tools from Markov decision processes (MDPs) play a central role in this line of research. It has been shown that the optimal transmission policy under the ZOH estimator exhibits a switching structure that depends \emph{only} on the instantaneous estimation error~\cite{luo2026exploiting, luo2025cost}. 

This work generalizes~\cite{luo2025cost} in several ways. First, we incorporate both the age and the semantics of information to improve estimation performance. Specifically, we employ the MAP estimator, which characterizes how the usefulness of outdated information evolves with AoI. To the best of our knowledge, the only prior work using the MAP estimator in semantics-aware communication systems is~\cite{cosandal2025joint}, which proposed a learning-based unstructured algorithm to minimize AoII in a pull-based system. In contrast, we study a push-based system and derive structural results for the optimal policy.

Second, we generalize the structural results in~\cite{luo2025cost} and show that the thresholds of the optimal switching policy depend on \emph{both} the instantaneous estimation error and the AoI. We present sufficient conditions under which the ZOH estimator yields the same performance as the MAP estimator.

Third, we impose a hard constraint on the transmission frequency. The resulting problem is a constrained MDP (CMDP), which is computationally more challenging to solve than the unconstrained MDP studied in~\cite{luo2025cost}. We fully characterize the structure of the constrained optimal policy and derive closed-form expressions for its parameters. In~\cite{luo2025cost}, a structured policy iteration (SPI) algorithm was introduced to compute the optimal switching policy. In this work, we develop Insec-SPI, which solves the CMDP in significantly fewer iterations than existing unstructured methods.

\begin{figure}[t!]
    \centering
    \scalebox{0.9}{\begin{tikzpicture}
    \draw[-{Latex[length=5pt, width=4pt]}] (-1.5,0) -- (6.8,0);
    
    \draw[thick] (-1, -0.15) -- (-1, 0.15);
    \draw[thick] (5.5, -0.15) -- (5.5, 0.15);
    
    \node[anchor=north] at (-1,-0.15) {$t$};
    \node[anchor=north] at (5.5,-0.15) {\small $t+1$};
    
    \filldraw (-0.6,0) circle (1.5pt);
    \filldraw (0.0,0) circle (1.5pt);
    \filldraw (0.8,0) circle (1.5pt);
    \filldraw (2,0) circle (1.5pt);
    \filldraw (4,0) circle (1.5pt);
    \filldraw (5,0) circle (1.5pt);
    \filldraw (6,0) circle (1.5pt);
    
    \node[anchor=south] at (-0.6,0) {$X_t$};
    \node[anchor=south] at (0.0,0) {$I_t$};
    \node[anchor=south] at (0.8,0) {$U_t$};
    \node[anchor=north] at (0.8,-0.05) {$\pi_t$};
    \node[anchor=south] at (2,0) {$(Z_t, \Theta_t)$};
    \node[anchor=south] at (4,0) {$(\hat{X}_t, \Delta_t)$};
    \node[anchor=north] at (4,-0.05) {$g_t$};
    \node[anchor=south] at (5,0) {$I^\prime_t$};
    \node[anchor=south] at (6.05,0) {$X_{t+1}$};
\end{tikzpicture}}
    \caption{Time ordering of the relevant variables.}
    \label{fig:timeline}
\end{figure}

\begin{figure*}[t!]
    \centering
    \subfloat[Positively correlated source, $M = (0.8, 0.2; 0.3, 0.7)$.]{
        \scalebox{0.8}{\begin{tikzpicture}
    \begin{axis}[
        axis lines = left,
        xlabel = {$\text{AoI}~\Theta_t$},
        ylabel = {$\text{Belief of}~X_t = 1$}, 
        legend pos = north east,
        domain = 0:12,        
        samples=13,
        samples at={0,...,12},
        xmin = 0, xmax = 12.5,
        ymin = 0, ymax = 1.1,
        width=10cm,           
        height=6cm,     
        grid = both,
        grid style={dashed,gray!30},  
        axis line style = {-Latex},   
        major grid style = {dashed,gray!30}, 
        tick label style={font=\small},  
        xtick = {0,2,...,12},
        ytick = {0,0.2,...,1.0},        
    ]
    \pgfmathsetmacro{\p}{0.2}
    \pgfmathsetmacro{\q}{0.3}
    
    \addplot[black, line width=0.45mm, dashed] {\p/(\p+\q) + \q/(\p+\q)*(1-\p-\q)^\x};
    \addlegendentry{$Z_t = 1$}

    \addplot[black, line width=0.4mm] {\p/(\p+\q) - \p/(\p+\q)*(1-\p-\q)^\x};
    \addlegendentry{$Z_t = 0$}
    \end{axis}
\end{tikzpicture}}
        \label{fig:AoCI_monotone}
    }
    \hspace{1cm}
    \subfloat[Negatively correlated source, $M = (0.2 , 0.8 ; 0.7 , 0.3)$.]{
        \scalebox{0.8}{\begin{tikzpicture}
    \begin{axis}[
        axis lines = left,
        xlabel = {$\text{AoI}~\Theta_t$},
        ylabel = {$\text{Belief of}~X_t = 1$}, 
        legend pos = north east,
        domain=0:12,
        samples=13,
        samples at={0,...,12},
        xmin = 0, xmax = 12.5,
        ymin = 0, ymax = 1.1,
        width=10cm,           
        height=6cm,     
        grid = both,
        grid style={dashed,gray!30},  
        axis line style = {-Latex},   
        major grid style = {dashed,gray!30}, 
        tick label style={font=\small},  
        xtick = {0,2,...,12},
        ytick = {0,0.2,...,1.0},        
    ]
    \pgfmathsetmacro{\p}{0.8}
    \pgfmathsetmacro{\q}{0.7}
    
    \addplot[black, line width=0.45mm, dashed] {\p/(\p+\q) + \q/(\p+\q)*(1-\p-\q)^\x};
    \addlegendentry{$Z_t = 1$}

    \addplot[black, line width=0.4mm] {\p/(\p+\q) - \p/(\p+\q)*(1-\p-\q)^\x};
    \addlegendentry{$Z_t = 0$}
    \end{axis}
\end{tikzpicture}}
        \label{fig:AoCI_nonmonotone}
    }
    \caption{Evolution of the belief value of a binary Markov chain. (a) The stationary distribution of the chain is $\mu = (0.6, 0.4)$. When the source was last observed in state $1$, i.e., $Z_t = 1$, the MAP estimate is $\hat{X}_t = 1$ if $0 \leq \Theta_t \leq 3$, and $\hat{X}_t = 0$ otherwise. When $Z_t = 0$, the MAP estimate remains $\hat{X}_t = Z_t = 0$ for all $\Theta_t \geq 0$. In (b), the MAP estimator enters a steady state when $Z_t = 0$ and $\Theta_t \geq 3$, or when $Z_t = 1$ and $\Theta_t \geq 4$.}
    \label{fig:AoCI}
\end{figure*}

\section{System Model and Problem Formulation}\label{sec:model}
\subsection{System Model}
Consider the remote estimation system depicted in Figure~\ref{fig:system_model}. The time ordering of the variables to be defined is illustrated in Figure~\ref{fig:timeline}. The stochastic process of interest is an irreducible finite-state Markov chain with alphabet $\mathbb{X}$ and transition matrix $M$, where $M_{i,j} = \Pr[X_{t+1} = j|X_t = i]$, $t\geq 1$, represents the probability of the source transitioning from state $i$ to state $j$ between two consecutive time slots. A chain is called symmetric if $M^\top = M$. The probability of being in state $j$ after $n$ consecutive time slots is given by the $(i,j)$-th entry of the $n$-th power of $M$, i.e., $\Pr[X_{t+n} = j|X_t = i] = M^n_{i,j}$. Given irreducibility, the limit $\lim_{n\to\infty}M^n$ exists and each row converges to the chain's stationary distribution $\mu$, where
\begin{equation*}
    \lim_{n\to\infty}M^n_{i,j} = \mu(j)\,\,\,\textrm{and}\,\,\, \mu M = \mu.
\end{equation*}

At each time $t$, the sensor selects an action $U_t\in\mathbb{U} = \{0, 1\}$, where it transmits a fresh measurement if $U_t = 1$, and remains silent otherwise. Let $H_t$ denote the channel state at time $t$, which follows an i.i.d. Bernoulli process with 
\begin{equation*}
    \Pr[H_t = 1] = p_s, ~\Pr[H_t = 0] = 1 - p_s = p_f.
\end{equation*}
Here, $H_t = 1$ indicates a successful transmission, while $H_t = 0$ denotes a transmission failure. The input alphabet of the channel is $\mathbb{X}$, and the output alphabet is $\mathbb{Y}= \mathbb{X}\cup\{\mathcal{E}\}$, where the symbol $\mathcal{E}$ denotes the event when no packet was received. The channel output $Y_t$ is given by
\begin{equation}
    Y_t = \begin{cases}
        X_t, &U_tH_t = 1,\\
        \mathcal{E}, &U_tH_t = 0.
    \end{cases}
\end{equation}

At each time $t$, the sensor has access to the current source state $X_t$ and one-step delayed channel feedback $Y_{t-1}$. Thus, the information available at the sensor up until time $t$ is 
\begin{equation*}
    I_t = (X_{1:t}, Y_{1:t-1}, U_{1:t-1}).
\end{equation*}
With this information, the sensor selects an action $U_t$ by a transmission rule $\pi_t : \mathbb{X}^{t}\times \mathbb{Y}^{t-1} \times \mathbb{U}^{t-1} \to \mathbb{U}$, so that
\begin{equation}
    U_t = \pi_t(I_t) = \pi_t(X_{1:t}, Y_{1:t-1}, U_{1:t-1}).\label{eq:trans-policy}
\end{equation}
The sequence $\pi = (\pi_1, \pi_2, \ldots)$ is called a \emph{transmission policy}. Let $\Pi$ denote the set of all admissible policies defined in~\eqref{eq:trans-policy}. Of particular interest are stationary policies, which select the same decision rule at all time steps. Stationary policies can be either \emph{deterministic}, where the action at each state is selected with certainty, or \emph{randomized}, where the action at each state is drawn according to a probability distribution.

The receiver employs the MAP estimation rule. At each time $t$ an estimate of $X_t$ is generated by the rule~\cite[Ch.~3.5]{krishnamurthy2016POMDP}
\begin{equation}
    \hat{X}_t = g_t(Y_{1:t}) = \argmax_{x\in \mathbb{X}}\Pr[X_t = x|Y_{1:t}].\label{eq:MAP-origin}
\end{equation}
The AoI $\Theta_t$ at the receiver is defined as the time elapsed since the newest update was generated, i.e.,
\begin{equation}
    \Theta_t := t - \max\{\tau \leq t: Y_\tau \neq \mathcal{E}\}. 
\end{equation}
Then, the receiver's belief of the source being in state $X_t = x$ can be written as a function of AoI as
\begin{equation}
    \Pr[X_t = x|Y_{1:t}] = \begin{cases}
        1, &\Theta_t = 0, Y_t = x,\\
        0, &\Theta_t = 0, Y_t \neq x,\\
        M^{\Theta_t}_{Z_t, x}, &\Theta_t \geq 1, Y_t = \mathcal{E},
    \end{cases}\label{eq:AoI}
\end{equation}
where $Z_t = X_{t - \Theta_t}$ is the latest update at the receiver. This implies that history need not be retained; instead, the receiver stores only the most recent update and its generation time. Figure~\ref{fig:AoCI} depicts the evolution of the belief $\Pr[X_t = 1|Y_{1:t}]$ as a function of $(Z_t, \Theta_t)$ for the cases of a positively correlated source and a negatively correlated one.

In the sequel, we express the MAP estimator as
\begin{equation}
    \hat{X}_t = g(Z_t, \Theta_t).\label{eq:MAP}
\end{equation}

\begin{definition}
    The MAP estimator is said to be in steady state if, given outdated information $X_{t-\Theta_t}$, the estimate $\hat{X}_t$ remains unchanged as the AoI $\Theta_t$ grows. 
\end{definition}

In the literature, AoI has been considered an inefficient metric for monitoring Markov chains~\cite{luo2025information}. However, this limitation applies only to ZOH estimators, i.e., $\hat{X}_t = Z_t$ for all $\Theta_t \geq 0$. For the MAP estimator, the usefulness of information at the receiver depends on both the age of information and its content. This dependence has been widely exploited in data scheduling over Markovian fading channels to facilitate analysis and computation (see, e.g.,~\cite{zhao2008myopic, liu2010indexability, ouyang2016downlink}), although these works did not explicitly interpret $\Theta_t$ as AoI until~\cite{costa2015age}.

\subsection{Performance Measure}
In classical estimation systems, the system's performance is measured by \textit{average distortion}. Given a bounded distortion function $d : \mathbb{X} \times \mathbb{X} \to [0, \infty)$, the average distortion of a policy $\pi$ is measured by
\begin{equation}
    \mathcal{J}_\textrm{classic}(\pi) := \limsup_{T\to\infty}\mathbb{E}^{\pi} \Bigg[
    \frac{1}{T}\sum_{t=1}^{T}d(X_t, \hat{X}_t)
    \Bigg].
\end{equation}
In semantics-aware systems, however, information accuracy is not the only concern. Let $\T{SoI}_t$ denote the semantic value conveyed by the current measurement $X_t$, defined as
\begin{equation}
    \T{SoI}_t := \eta_t(I^\prime_t) = \eta_t(X_{1:t}, Y_{1:t}, U_{1:t}),
\end{equation}
where $\eta_t: \mathbb{X}^{t} \times \mathbb{Y}^{t} \times \mathbb{U}^{t} \to \Gamma$ is an extraction function that quantifies data significance through the history of all system realizations prior to time $t+1$, and $\Gamma \subseteq [0, \infty)$ is a real-valued domain. The system's performance is measured by
\begin{equation}
    \mathcal{J}(\pi) := \limsup_{T\to\infty}\mathbb{E}^{\pi} 
    \Bigg[
    \frac{1}{T}\sum_{t=1}^{T}c(X_t, \hat{X}_t, \T{SoI}_t)
    \Bigg],\label{eq:semantics-aware average cost}
\end{equation}
where $\hat{X}_t = g(Z_t, \Theta_t)$, and $c: \mathbb{X} \times \mathbb{X} \times \Gamma \to [0, \infty)$ is a cost function, incorporating both the estimation error $(X_t, \hat{X}_t)$ and its semantic value $\T{SoI}_t$. Notably, unlike the distortion function $d(X_t, \hat{X}_t)$, which depends solely on instantaneous estimation error, the semantic attribute $\T{SoI}_t$ may be \emph{history-dependent}. 

We instantiate the semantic attribute as the \emph{urgency of lasting impact} in consecutive errors~\cite{luo2025cost}. That is, the longer an error persists, the more severe its consequences can become. The error holding time is given by
\begin{equation}
    \Delta_t \hspace{-0.2em}=\hspace{-0.2em} \begin{cases}
        \Delta_{t-1} \hspace{-0.2em}+\hspace{-0.2em} 1, &X_t\neq \hat{X}_t, (X_t, \hat{X}_t) \hspace{-0.2em}=\hspace{-0.2em} (X_{t-1}, \hat{X}_{t-1}),\\
        1, &X_t\neq \hat{X}_t, (X_t, \hat{X}_t) \hspace{-0.2em}\neq\hspace{-0.2em} (X_{t-1}, \hat{X}_{t-1}),\\
        0, &X_t = \hat{X}_t.
    \end{cases}
\end{equation}
The urgency of lasting impact is defined as
\begin{equation}
    c(X_t, \hat{X}_t, \Delta_t) = \rho_{X_t, \hat{X}_t}(\Delta_t),\label{eq:semantics-aware cost}
\end{equation}
where $\rho_{i,j}$ are non-decreasing, possibly unbounded functions that impose escalating penalties for prolonged errors. These functions are general and may be discontinuous or non-convex.

The semantic attribute~\eqref{eq:semantics-aware cost} captures the contextual severity of estimation errors at the transmitter, while the information age characterizes the predictability of outdated data at the receiver through the MAP estimator~\eqref{eq:MAP}. By jointly incorporating information age and semantics, the proposed framework enables more informed decisions in resource-constrained systems, \textit{avoiding unnecessary transmissions triggered by predictable or low-impact errors.} This highlights the complementary roles of age and semantics in remote estimation.

\subsection{Problem Formulation}
The goal is to minimize the estimation cost $\mathcal{J}(\pi)$ subject to a hard constraint on the transmission frequency $F(\pi)$, where
\begin{equation}
    F(\pi) := \limsup_{T\to\infty}\mathbb{E}^{\pi} 
    \Bigg[
    \frac{1}{T}\sum_{t=1}^{T} U_t
    \Bigg].
\end{equation}
This problem is formulated as
\begin{equation}
    \inf_{\pi\in\Pi}\mathcal{J}(\pi),\,\,\,\textrm{subject to}\,\,\, F(\pi) \leq F_{\max}, \label{problem:constrained-semantics-aware problem}
\end{equation}
where $\Pi$ is the set of all admissible policies, and $F_{\max} \in (0, 1)$ is the maximum allowed transmission frequency.

\begin{definition}
    A policy $\pi$ is feasible if it satisfies $\mathcal{J}(\pi)<\infty$ and $F(\pi)\leq F_{\max}$. Let $\Pi_\textrm{f} \subset \Pi$ denote the set of all feasible policies. A policy is said to be (constrained) optimal if it attains the minimum $\mathcal{J}^* = \inf_{\pi\in\Pi_{\textrm{f}}}\mathcal{J}(\pi)$.
\end{definition}

Problem~\eqref{problem:constrained-semantics-aware problem} is a CMDP. Due to intermittent transmissions and channel unreliability, the cost function~\eqref{eq:semantics-aware cost} can grow indefinitely. This violates the boundedness assumption typically imposed on distortion functions, making the semantics-aware estimation problem inherently more challenging than its classic counterpart. 

A formal approach to CMDPs is the Lagrange multiplier method. Let $\lambda\geq0$ be a multiplier, which can be interpreted as the communication cost associated with each transmission attempt. Let 
\begin{equation}
    \ell(I_t, U_t) = \mathbb{E}[c(X_t, \hat{X}_t, \Delta_t)|I_t, U_t] + \lambda U_t
\end{equation}
denote the estimation and communication costs of taking action $U_t$. Define the long-run average Lagrangian cost as
\begin{align}
     \mathcal{L}^\lambda(\pi) 
     &:= \limsup_{T\to\infty}\mathbb{E}^{\pi} 
    \Bigg[
    \frac{1}{T}\sum_{t=1}^{T}\ell(I_t, U_t)
    \Bigg]\notag\\
    & =\mathcal{J}(\pi) + \lambda F(\pi).\label{eq:lagrange cost}
\end{align}
The Lagrangian approach then seeks to solve the following unconstrained MDP
\begin{equation}
    \inf_{\pi \in \Pi} \mathcal{L}^\lambda(\pi).\label{problem:MDP}
\end{equation}

\begin{definition}
    For a given $\lambda$, a policy $\pi_\lambda \in \Pi$ is said to be $\lambda$-optimal if it attains the minimum in~\eqref{problem:MDP}. For notational simplicity, we shall write $\mathcal{J}^\lambda = \mathcal{J}(\pi_\lambda)$, $F^\lambda = F(\pi_\lambda)$ and $\mathcal{L}^\lambda = \mathcal{L}^\lambda(\pi_\lambda) = \mathcal{J}^\lambda + \lambda F^\lambda$.
\end{definition}

Note that a $\lambda$-optimal policy to Problem~\eqref{problem:MDP} need not be feasible or optimal for the constrained Problem~\eqref{problem:constrained-semantics-aware problem}. The following result (e.g., Beutler and Ross \cite{ross1985CMDP}) gives conditions for a $\lambda$-optimal policy to be optimal.
\begin{proposition}\label{proposition:ross}
A $\lambda$-optimal policy solves Problem~\eqref{problem:constrained-semantics-aware problem} if:
\begin{enumerate}
    \item[i.] $\lambda = 0$ and $\pi_0$ satisfies $\mathcal{J}^0 < \infty$ and $F^0 \leq F_{\max}$; or
    \item[ii.] $\lambda > 0$ and $\pi_\lambda$ satisfies $\mathcal{J}^\lambda < \infty$ and $F^\lambda = F_{\max}$.
\end{enumerate}
\end{proposition}
The first condition can be identified directly. In general, however, we shall be dealing with the challenging Condition (ii). The key conceptual difficulty lies in meeting the equality $F^\lambda = F_{\max}$, which necessitates careful calibration of $\lambda$.

The following monotonicity properties (e.g., Luo and Pappas~\cite[proposition~3]{luo2025semantic}) will be instrumental in characterizing and searching for an optimal policy. Intuitively, as the value of $\lambda$ (communication cost) increases, it would be wise to transmit less frequently (i.e., a smaller $F^\lambda$), which in turn degrades the estimation performance (i.e., a larger $\mathcal{J}^\lambda$).

\begin{proposition}\label{proposition:monotonicity}
    The function $F^\lambda$ ($\mathcal{J}^\lambda$) is piecewise constant and non-increasing (non-decreasing) in $\lambda$. The Lagrangian cost $\mathcal{L}^\lambda$ is continuous, piecewise-linear and concave.
\end{proposition}

\section{Structural Results}\label{sec:main-results}
This section presents the structural results for the optimal policy. We first establish the existence of a \emph{deterministic} $\lambda$-optimal policy for any given multiplier $\lambda$. Such a $\lambda$-optimal policy admits a \emph{switching structure} that triggers transmission only when the error duration exceeds a threshold that depends on both the AoI and the instantaneous error. We then identify conditions under which the thresholds are independent of the AoI, implying that the MAP and ZOH estimators coincide.

Next, we show the existence of an optimal multiplier $\lambda^*$ and a \emph{randomized} $\lambda^*$-optimal policy that solves the constrained problem. Notably, such an optimal policy has a \emph{simple mixture structure} that randomizes between two neighboring switching policies with a fixed probability, and we derive a closed-form expression for the optimal mixture coefficient. 

\subsection{Existence of an Optimal Policy}
Recall that the cost function~\eqref{eq:semantics-aware cost} is (possibly) unbounded. This naturally raises the question of feasibility for Problems~\eqref{problem:constrained-semantics-aware problem} and~\eqref{problem:MDP}. In the following, we first give conditions under which a $\lambda$-optimal policy exists for every given $\lambda$. Then we show the existence of a feasible $\lambda$-optimal policy. These results are embodied in Theorem~\ref{theorem:existence}.

\begin{assumption}\label{assumption:existence}
    The nonlinear functions $\rho_{i,j}$ satisfy
    \begin{equation}
        \lim_{\delta\to\infty}\frac{\rho_{i,j}(\delta+1)}{\rho_{i,j}(\delta)} < \frac{1}{M_{i,i}p_f}, \quad i\neq j,
    \end{equation}
    where $M_{i,i}p_f$ is the error persistence probability in state $i$.
\end{assumption}

\begin{theorem}\label{theorem:existence}
    Let Assumption~\ref{assumption:existence} hold. The following are true:
    \begin{enumerate}
        \item[i.] For any $\lambda \geq 0$, there exists a $\lambda$-optimal policy that achieves finite average cost, i.e., $\mathcal{J}^\lambda <\infty$ for all $\lambda$.
        \item[ii.] There exists some $\lambda\geq 0$ for which the corresponding $\lambda$-optimal policy is feasible, i.e., $F^\lambda \leq F_{\max}$.
    \end{enumerate}
\end{theorem}
\begin{IEEEproof}
    See Appendix~\ref{proof:theorem existence}.
\end{IEEEproof}

\subsection{Structure of \texorpdfstring{$\lambda$}{lambda}-optimal Policies}
We first show that the sensor can discard most of its historical information without losing optimality.
\begin{lemma}\label{lemma:information-state}
    There is no loss of optimality in restricting attention to transmission rules of the form
    \begin{equation}
        U_t = \pi_t(X_t, Z_{t-1}, \Theta_{t-1}, \Delta_{t-1}).
    \end{equation}
    The \emph{information state} (i.e., a sufficient statistic for decision-making) is given by
    \begin{equation}
        S_t := (X_t, Z_{t-1}, \Theta_{t-1}, \Delta_{t-1}).\label{eq:information-state}
    \end{equation}
\end{lemma}
Let $\mathbb{S}=\mathbb{X}^2\times\mathbb{N}^2$ denote the state space of $\{S_t\}$. In the sequel, we shall write $U_t = \pi_t(S_t)$ and $\ell(I_t,U_t)=\ell(S_t,U_t)$.

\begin{IEEEproof}
    We need to show that $\{S_t\}$ is a controlled Markov process controlled by $\{U_t\}$. See Appendix~\ref{proof:lemma-information-state}.
\end{IEEEproof}

The next result is a direct consequence of Theorem~\ref{theorem:existence} and Lemma~\ref{lemma:information-state} (e.g., Puterman~\cite[Ch.~8.10]{puterman1994markov}). It implies that, for the unconstrained Problem~\eqref{problem:MDP}, there is no loss of optimality in restricting attention to stationary \emph{deterministic policies}, and the $\lambda$-optimal policy can be obtained by solving the Bellman equation recursively.

\begin{proposition}
    The $\lambda$-optimal policy $\pi_\lambda$ is stationary deterministic such that for every $s\in\mathbb{S}$, $u = \pi_\lambda(s)$ solves the following Bellman equation
    \begin{equation}
        \mathcal{L}^\lambda + V(s) = \min_{u}
        \big\{
         \ell(s, u) + \mathbb{E}[V(s^\prime) |s, u]
        \big\},\label{eq:bellman-equation}
    \end{equation}
    where $V(s)$ is a bounded function, and $s^\prime = (x^\prime, z^\prime, \theta^\prime, \delta^\prime)$ is the next state after taking action $u$. The Bellman equation~\eqref{eq:bellman-equation} can be solved using the recursion
    \begin{subequations}
        \begin{align}
    Q^n(s,u) &= \ell(s,u) + \mathbb{E}[V^{n-1}(s^\prime)|s,u],\label{eq:Q-factor-1}\\
    \tilde{V}^n(s) &= \min_{u}\{Q^n(s,u)\},\label{eq:Q-factor-2}\\
    V^n(s) &= \tilde{V}^n(s) - \tilde{V}^n(s_\textrm{ref}),\label{eq:Q-factor-3}
\end{align}\label{eq:Q-factor}
    \end{subequations}
where $Q^n(s,u)$ is the Q-factor at iteration $n$, and $s_\textrm{ref}\in\mathbb{S}$ is an arbitrary reference state. The sequences $\{Q^n(s, u)\}$, $\{\tilde{V}^n(s)\}$ and $\{V^n(s)\}$, converge as $n \to \infty$. Moreover, $V(s) = \tilde{V}(s) - \tilde{V}(s_\textrm{ref})$ and $\mathcal{L}^\lambda = \tilde{V}(s_\textrm{ref})$ are a solution to \eqref{eq:bellman-equation}.
\end{proposition}

So far, we have shown that history need not be retained and nonrandomized policies suffice for Problem~\eqref{problem:MDP}. This simplifies computation through reduced storage and fewer arithmetic operations. One can use classical dynamic programming methods~\cite{puterman1994markov} to solve the Bellman equation~\eqref{eq:bellman-equation}. However, they often entail substantial computational overhead.

In the next theorem, we show that it suffices to restrict attention to \emph{switching policies}. This structural result significantly reduces computational and memory complexity. To implement a switching policy, one only needs to compute and store $|\mathbb{X}|(|\mathbb{X}|-1)$ thresholds, rather than computing the optimal action for every state. We first introduce the following definitions.

\begin{definition}
    A policy $\pi$ is said to have a \emph{switching} structure if, given error $(x, z)$ and AoI $\theta$, $\pi(x, z, \theta, \delta) = 1$ only when the error holding time $\delta$ exceeds a threshold $\delta^*_{x, g(z, \theta)}$. Formally,
    \begin{equation}
        \pi(x, z, \theta, \delta) = \begin{cases}
            1, &x\neq g(z, \theta), \delta\geq \delta^*_{x, g(z, \theta)},\\
            0, &\textrm{otherwise}.
        \end{cases}\label{eq:switching-structure}
    \end{equation}
    A \emph{threshold} policy is a switching policy with identical thresholds across all errors, i.e., $\delta^*_{x, g(z, \theta)} = \delta^*$ for all $x\neq g(z, \theta)$.
\end{definition}

\begin{definition}\label{eq:poset}
   Let $(\mathbb{S}, \preceq)$ be a partially ordered set. For two elements $s_1, s_2 \in \mathbb{S}$, we define
    \begin{equation*}
        s_1 \preceq s_2 \iff (x_1, z_1, \theta_1) = (x_2, z_2, \theta_2)~\textrm{and}~\delta_1 \leq \delta_2.
    \end{equation*}
    If $(x_1, z_1, \theta_1) \neq (x_2, z_2, \theta_2)$, then $s_1$ and $s_2$ are not comparable. A function $f: \mathbb{S} \to \mathbb{R}$ is \emph{monotone} if
    \begin{equation*}
        s_1 \preceq s_2 \implies f(s_1) \leq f(s_2).
    \end{equation*}
\end{definition}

\begin{theorem}\label{theorem:lambda-optimal}
    Switching policies suffice for Problem~\eqref{problem:MDP}.
\end{theorem}
\begin{IEEEproof}
    (Sketch) The key step is to show that the value function $V(s)$ is monotone. It then follows that the $\lambda$-optimal policy $\pi_\lambda(s)$ is also monotone and therefore has a switching structure. Detailed arguments are provided in Appendix~\ref{proof:theorem-lambda-optimal}.
\end{IEEEproof}

Theorem~\ref{theorem:lambda-optimal} does not imply that for a fixed $\lambda$, every $\lambda$-optimal policy necessarily has a switching structure. Rather, it guarantees the existence of a switching $\lambda$-optimal policy. As shown later in Theorem~\ref{theorem:optimal-policy}, a $\lambda$-optimal policy that solves the original constrained problem may be randomized.

The following corollaries present conditions under which the thresholds of the $\lambda$-optimal policy are independent of AoI or estimation error. Under these conditions, the ZOH estimator performs as well as the MAP estimator. Proofs are provided in Appendix~\ref{proof:coro symmetric source}.

\begin{corollary}\label{coro:symmetric-source}
    Assume $M$ is symmetric and satisfies
    \begin{equation}
        M_{i,j} = \begin{cases}
            \sigma, &i\neq j,\\
            1 - (|\mathbb{X}|-1)\sigma, &i=j,
        \end{cases}\label{eq:special-source}
    \end{equation}
    and $\sigma \leq 1/|\mathbb{X}|$. Then, the information state is independent of AoI, i.e., $S_t = (X_t, Z_{t-1}, \Delta_{t-1})$. The $\lambda$-optimal policy satisfies
    \begin{equation}
        \pi_\lambda(x, z, \theta, \delta) = \pi_\lambda(x, z, \delta) = \begin{cases}
            1, &x\neq z, \delta\geq \delta^*_{x, z},\\
            0, &\textrm{otherwise}.
        \end{cases}
    \end{equation}
\end{corollary}

\begin{corollary}\label{coro:special-case}
    Assume $M$ has the form~\eqref{eq:special-source} and $\rho_{i,j} = \rho$ for all $i\neq j$. Then, the information state is independent of AoI and the instantaneous estimation error, i.e., $S_t = \Delta_{t-1}$. The $\lambda$-optimal policy degenerates into a threshold policy, i.e.,
    \begin{equation}
        \pi_\lambda(x, z, \theta, \delta) = \pi_\lambda(\delta) = \begin{cases}
            1, &\delta\geq \delta^*,\\
            0, &\textrm{otherwise}.
        \end{cases}
    \end{equation}
\end{corollary}

\subsection{Structure of Optimal Policies}
This section establishes structural properties of the constrained optimal policy. We begin by introducing a class of simple mixture policies.

\begin{definition}
    Let $\alpha \in (0,1)$ denote a mixture coefficient. A randomized policy $(\alpha, \pi_1, \pi_2)$ is called a \emph{mixture policy} if, at each time, it selects $\pi_1$ with probability $\alpha$ and $\pi_2$ with probability $1-\alpha$. It is called a \emph{simple mixture policy} if:
    \begin{itemize}
        \item $\pi_1$ and $\pi_2$ are deterministic policies, and the chains $\{S_t\}$ induced by $\pi_1$ and $\pi_2$ are positive recurrent;
        \item $\pi_1$ and $\pi_2$ differ in exactly one state $i_0 \in \mathbb{S}$, i.e., $\pi_1(s)=\pi_2(s)$ for all $s \neq i_0$ and $\pi_1(i_0)\neq \pi_2(i_0)$;
        \item randomization occurs only when the system visits $i_0$, and it follows the selected policy until the next visit to $i_0$.
    \end{itemize}
\end{definition}

The following theorem establishes the existence of an optimal simple mixture policy and provides a closed-form expression for the mixture coefficient.

\begin{theorem}\label{theorem:optimal-policy}
    The following hold for an optimal policy:
    \begin{enumerate}
        \item[i.] Suppose $F^\lambda = F_{\max}$ for some $\lambda$. The switching policy $\pi_\lambda$ of the form \eqref{eq:switching-structure} is optimal.
        \item[ii.] Suppose no such $\lambda$ as above exists. Then there exist $\lambda>0$ and an arbitrarily small $\epsilon>0$ such that $F^{\lambda + \epsilon} < F_{\max} < F^{\lambda - \epsilon}$, and $\pi_{\lambda+\epsilon}$ and $\pi_{\lambda-\epsilon}$ differ in exactly one state, denoted by $i_0$. Moreover, the simple mixture policy $\pi^*_\alpha = (\alpha, \pi_{\lambda-\epsilon}, \pi_{\lambda+\epsilon})$ is optimal, where the mixture coefficient $\alpha \in (0, 1)$ is given by\footnote{A commonly used approximation is $\tilde{\pi} = (\beta,\pi_{\lambda-\epsilon},\pi_{\lambda+\epsilon})$~\cite{luo2025semantic}.}
        \begin{equation}
            \alpha = \frac{\beta \mu_{\pi_{\lambda - \epsilon}}(i_0)}{\beta \mu_{\pi_{\lambda - \epsilon}}(i_0) + (1-\beta) \mu_{\pi_{\lambda + \epsilon}}(i_0)}, \label{eq:mixture coefficient}
        \end{equation}
        where $\beta = (F_{\max} - F^{\lambda + \epsilon})/ (F^{\lambda - \epsilon} - F^{\lambda + \epsilon})$, and $\mu_{\pi_{\lambda-\epsilon}}$ and $\mu_{\pi_{\lambda+\epsilon}}$ are the stationary distributions of $\{S_t\}$ induced by $\pi_{\lambda-\epsilon}$ and $\pi_{\lambda+\epsilon}$, respectively.
    \end{enumerate}
\end{theorem}
\begin{IEEEproof}
Part (i) holds trivially. Now assume that $F^\lambda \neq F_{\max}$ for all $\lambda$. From Proposition~\ref{proposition:monotonicity}, we know that $F^\lambda$ is piecewise constant and non-increasing. Hence, there exists a breakpoint $\lambda$ at which $F^{\lambda + \epsilon} < F_{\max} < F^{\lambda - \epsilon}$ for arbitrarily small $\epsilon>0$. Given the existence of stationary and deterministic $\lambda$-optimal policies for the unconstrained problem~\eqref{problem:MDP}, the constrained optimal policy is a randomized mixture of two deterministic policies that differ in exactly one state (see, e.g., Beutler and Ross \cite{ross1985CMDP}). 

It remains to show that the simple mixture policy $\pi^*_\alpha$ with mixture coefficient~\eqref{eq:mixture coefficient} is optimal. It can be shown that the chains $\{S_t\}$ induced by $\pi_{\lambda-\epsilon}$ and $\pi_{\lambda+\epsilon}$ are positive recurrent, and that $i_0$ is a positive recurrent state under both policies. This implies that the chain $\{S_t\}$ induced by $\pi^*_\alpha$ has a single recurrent class, and every state in this class has a finite expected return time to $i_0$. Therefore, $\pi^*_\alpha$ induces a regenerative structure with regeneration state $i_0$.

Let $T_c$ denote the regeneration cycle length (i.e., the time between successive visits to $i_0$), and let $N_c(s)$ denote the number of times it visits state $s$ during one regeneration cycle. Because the choice between $\pi_{\lambda-\epsilon}$ and $\pi_{\lambda+\epsilon}$ is made only at regeneration epochs, conditioning on this choice yields
\begin{align*}
    \mathbb{E}^{\pi^*_\alpha}[T_c]
    &= \alpha \mathbb{E}^{\pi_{\lambda-\epsilon}}[T_c] + (1-\alpha)\mathbb{E}^{\pi_{\lambda+\epsilon}}[T_c], \\
    \mathbb{E}^{\pi^*_\alpha}[N_c(s)]
    &= \alpha \mathbb{E}^{\pi_{\lambda-\epsilon}}[N_c(s)] + (1-\alpha)\mathbb{E}^{\pi_{\lambda+\epsilon}}[N_c(s)]. 
\end{align*}
Since $\{S_t\}$ induced by $\pi^*_\alpha$ and $\pi_{\lambda\pm \epsilon}$ are positively recurrent, the stationary distributions $\mu_{\pi^*_\alpha}$ and $\mu_{\pi_{\lambda \pm \epsilon}}$ satisfy
\begin{equation*}
    \mu_{\pi^*_\alpha}(s) = \frac{\mathbb{E}^{\pi^*_\alpha}[N_c(s)]}{\mathbb{E}^{\pi^*_\alpha}[T_c]},\,\, \mu_{\pi_{\lambda \pm \epsilon}}(s) = \frac{\mathbb{E}^{\pi_{\lambda \pm \epsilon}}[N_c(s)]}{\mathbb{E}^{\pi_{\lambda \pm \epsilon}}[T_c]}.
\end{equation*}
In particular, $\mu_{\pi_{\lambda \pm \epsilon}}(i_0) = 1/\mathbb{E}^{\pi_{\lambda \pm \epsilon}}[T_c]$. 

Then we may write
\begin{align*}
    \mu_{\pi^*_\alpha}(s)
    &= \frac{\alpha \mathbb{E}^{\pi_{\lambda - \epsilon}}[N_c(s)] + (1-\alpha)\mathbb{E}^{\pi_{\lambda + \epsilon}}[N_c(s)]}{\alpha \mathbb{E}^{\pi_{\lambda - \epsilon}}[T_c] + (1-\alpha)\mathbb{E}^{\pi_{\lambda + \epsilon}}[T_c]}\\
    &=
    \frac{\alpha \mu_{\pi_{\lambda - \epsilon}}(s)\mathbb{E}^{\pi_{\lambda - \epsilon}}[T_c]
    + (1-\alpha)\mu_{\pi_{\lambda + \epsilon}}(s)\mathbb{E}^{\pi_{\lambda + \epsilon}}[T_c]}
    {\alpha \mathbb{E}^{\pi_{\lambda - \epsilon}}[T_c] + (1-\alpha)\mathbb{E}^{\pi_{\lambda + \epsilon}}[T_c]}\\
    &= \beta \mu_{\pi_{\lambda - \epsilon}}(s) + (1-\beta)\mu_{\pi_{\lambda + \epsilon}}(s),
\end{align*}
where $\beta$ is the effective long-run weight of $\pi_{\lambda - \epsilon}$, given by
\begin{align*}
    \beta
    &=
    \frac{\alpha/\mathbb{E}^{\pi_{\lambda + \epsilon}}[T_c]}
    {\alpha/\mathbb{E}^{\pi_{\lambda + \epsilon}}[T_c] + (1-\alpha)/\mathbb{E}^{\pi_{\lambda - \epsilon}}[T_c]}\\
    &=
    \frac{\alpha \mu_{\pi_{\lambda + \epsilon}}(i_0)}
    {\alpha \mu_{\pi_{\lambda + \epsilon}}(i_0) + (1-\alpha)\mu_{\pi_{\lambda - \epsilon}}(i_0)}.
\end{align*}
It follows that
\begin{align}
    \mathcal{J}(\pi^*_\alpha) 
    &= \sum_{s} \mu_{\pi^*_\alpha}(s) c(s)\notag \\
    &= \sum_{s}\big(\beta \mu_{\pi_{\lambda - \epsilon}}(s) + (1-\beta) \mu_{\pi_{\lambda + \epsilon}}(s)\big) c(s)\notag \\
    &= \beta \mathcal{J}^{\lambda - \epsilon} + (1-\beta) \mathcal{J}^{\lambda + \epsilon}.\label{eq:linearity-1}
\end{align}
Similarly, we obtain
\begin{align}
    F(\pi^*_\alpha) = \beta F^{\lambda - \epsilon} + (1-\beta) F^{\lambda + \epsilon}.\label{eq:linearity-2}
\end{align}
Combining~\eqref{eq:linearity-1} and~\eqref{eq:linearity-2} gives
\begin{align*}
    \mathcal{L}^\lambda(\pi^*_\alpha) 
    &= \mathcal{J}(\pi^*_\alpha) + \lambda F(\pi^*_\alpha)\\
    &= \big(\beta \mathcal{J}^{\lambda - \epsilon} + (1-\beta) \mathcal{J}^{\lambda + \epsilon}\big) \\
    &\quad+ \lambda \big(\beta F^{\lambda - \epsilon} + (1-\beta) F^{\lambda + \epsilon}\big)\\
    &= \beta \big(\mathcal{J}^{\lambda-\epsilon} + (\lambda-\epsilon) F^{\lambda-\epsilon}\big) \\
    &\quad + 
    (1-\beta) \big(\mathcal{J}^{\lambda+\epsilon} + (\lambda+\epsilon) F^{\lambda+\epsilon}\big) + o(\epsilon)\\
    &= \beta \mathcal{L}^{\lambda-\epsilon} + (1-\beta) \mathcal{L}^{\lambda+\epsilon} + o(\epsilon),
\end{align*}
where $o(\epsilon) = \epsilon \beta F^{\lambda-\epsilon} - \epsilon (1-\beta) F^{\lambda+\epsilon}$. Recall that $\mathcal{L}^\lambda$ is continuous in $\lambda$. Taking the limit as $\epsilon \to 0^+$ yields
\begin{equation}
    \mathcal{L}^\lambda(\pi^*_\alpha) = \mathcal{L}^\lambda,\quad \alpha \in (0,1),
\end{equation}
which implies that $\pi^*_\alpha = (\alpha, \pi_{\lambda-\epsilon}, \pi_{\lambda + \epsilon})$ is $\lambda$-optimal for all $\alpha\in (0, 1)$. By choosing $\alpha$ to satisfy~\eqref{eq:mixture coefficient}, we obtain $F(\pi^*_\alpha) = F_{\max}$. It follows from Proposition~\ref{proposition:ross}.ii that $\pi^*_\alpha$ is constrained optimal. This concludes our proof.
\end{IEEEproof}

The two neighboring policies $\pi_{\lambda-\epsilon}$ and $\pi_{\lambda+\epsilon}$ differ in a single \emph{information state}. Since both are switching policies, they differ in exactly one threshold. In the special case of Corollary~\ref{coro:special-case}, the optimal policy is a mixture of two threshold policies: one with threshold $\delta^*$ and the other with threshold $\delta^* + 1$. Moreover, the inequality $F^{\lambda + \epsilon} < F_{\max} < F^{\lambda - \epsilon}$ implies that $\pi_{\lambda+\epsilon}$ is feasible, whereas $\pi_{\lambda-\epsilon}$ is infeasible. 

\section{Structure-Aware Algorithm}\label{sec:computation}
This section develops a structure-aware algorithm, Insec-SPI, for computing the optimal policy, as depicted in Figure~\ref{fig:algorithm_diagram}. The algorithm comprises two key components: 
\begin{itemize}
    \item \textit{Structured Policy Iteration (SPI):} It leverages the switching structure (see Theorem~\ref{theorem:lambda-optimal}) to compute a $\lambda_n$-optimal policy at each iteration $n$, i.e., $\pi_{\lambda_n} = \texttt{SPI}(\lambda_n)$.
    \item \textit{Intersection Search (Insec):} It leverages the monotonicity properties in Proposition~\ref{proposition:monotonicity} to update the Lagrange multiplier, i.e., $\lambda_{n+1} = \texttt{Insec}(\lambda_n)$.
\end{itemize}

The algorithm converges when the optimal multiplier $\lambda^*$ is found. The optimal simple mixture policy is then constructed according to Theorem~\ref{theorem:optimal-policy}. The SPI and Insec-SPI algorithms are summarized in Algorithms~\ref{alg:SPI} and~\ref{alg:full}, respectively.

\begin{figure}[t!]
    \centering
    \scalebox{0.85}{\begin{tikzpicture}[
    node distance=1.5cm,
    box/.style={rectangle, draw, minimum width=2.5cm, minimum height=1cm, align=center},
    dashed box/.style={rectangle, draw, dashed, minimum width=8cm, minimum height=3cm}
]

\node [box, thick, rounded corners=1.5pt, rectangle, minimum width=1.5cm, minimum height=0.8cm] (spi) {SPI\\ (Theorem~\ref{theorem:lambda-optimal})};
\node [box, thick, rounded corners=1.5pt, rectangle, right=1.4cm of spi, minimum width=2cm, minimum height=0.8cm] (lagrange) {Insec Search\\
(Proposition~\ref{proposition:monotonicity})};
\node[box, thick, rounded corners=1.5pt, rectangle, right=1.4cm of lagrange, minimum width=2cm, minimum height=0.8cm] (lam) {Mixture Policy\\(Theorem~\ref{theorem:optimal-policy})};

\draw[->, thick] (spi) -- node[above] {$\pi_{\lambda_n}$} (lagrange);
\draw[->, thick] (lagrange) -- ++(0, -1.) -| node[near start, below] {$\lambda_{n+1}$} (spi);
\draw[->, thick] (lagrange) -- node[above] {$\lambda^*$} coordinate[midway] (mid-cd) (lam);

\end{tikzpicture}}
    \caption{Schematic representation of the Insec-SPI algorithm.}
    \label{fig:algorithm_diagram}
\end{figure}

\begin{algorithm}[t!]
\renewcommand{\algorithmicrequire}{\textbf{Input:}}
\renewcommand{\algorithmicensure}{\textbf{Output:}}
\caption{\texttt{SPI}$(\lambda)$}
\label{alg:SPI}
\begin{algorithmic}[1]
\Require $\lambda$, $s_{\textrm{ref}}$, $\Theta_{\max}, \Delta_{\max}$
\Ensure $\lambda$-optimal policy $\pi_\lambda$
\State \textsc{Initialize:} $V^0 \gets 0$, $\pi_\lambda^{0}(s) \gets \infty$ for all $x \neq g(z, \theta)$
\For{$n = 0, 1, \ldots$}
    \State \textsc{Find} a scalar $\set{L}^n$ and a vector $V^n$ by solving 
    \State \quad $\mathcal{L}^n + V^n(s)
        = \ell(s, \pi_\lambda^{n}(s))
        + \mathbb{E}[V^n(s^\prime)| s, \pi_\lambda^{n}(s)]$
    \State for all $s \in \tilde{\mathbb{S}}$ subject to $V^n(s_{\textrm{ref}}) = 0$.
    \For{each error $(x, z, \theta)$} 
        \For{$\delta = 1, \ldots, \Delta_{\max}$}
            \State $s \gets (x, z, \theta, \delta)$
            \State $Q(s, u) \gets \ell(s,u)
                    + \mathbb{E}[V^n(s^\prime)| s, u]$
            \State $\pi_\lambda^{n+1}(s)
                \gets \argmin_{u} Q(s, u)$
            \If{$\pi_\lambda^{n+1}(s) = 1$}
                \State $\pi_\lambda^{n+1}(s^\prime) \gets 1$ for all $s' \succeq s$
                \State \textbf{break}
            \EndIf
        \EndFor
    \EndFor
    \If{$\pi_\lambda^{n+1} = \pi_\lambda^{n}$}
        \State \Return $\pi_\lambda \gets \pi_\lambda^{n}$
    \EndIf
\EndFor
\end{algorithmic}
\end{algorithm}

\begin{algorithm}[ht!]
\renewcommand{\algorithmicrequire}{\textbf{Input:}}
\renewcommand{\algorithmicensure}{\textbf{Output:}}
\caption{Insec-SPI}
\label{alg:full}
\begin{algorithmic}[1]
\Require $F_{\max}$, $\lambda_{\max}$
\Ensure optimal policy $\pi^*$
\State \textsc{Initialize:} $\lambda_0^- \gets 0$, $\lambda_0^+ \gets \lambda_{\max}$
\If{$F^0 \leq F_{\max}$} \Comment{Check if $\lambda^* = 0$}
    \State \Return $\pi^* \gets \texttt{SPI}(0)$
\EndIf
\For{$n = 1, 2, \dots$} \Comment{Find breakpoint $\lambda^* > 0$}
    \State \textsc{Update} multiplier by intersection search:
    \State \quad $\lambda_n \gets (\mathcal{J}^{\lambda^+_{n-1}} - \mathcal{J}^{\lambda^-_{n-1}})/(F^{\lambda^-_{n-1}} - F^{\lambda^+_{n-1}})$
    \State \quad $\tilde{\mathcal{L}}^{\lambda_n} \gets \mathcal{J}^{\lambda^-_{n-1}} + \lambda_n F^{\lambda^-_{n-1}}$
    \State \quad $\pi_{\lambda_n} \gets \texttt{SPI}(\lambda_n)$
    \If{$\mathcal{L}^{\lambda_n} = \tilde{\mathcal{L}}^{\lambda_n}$} \Comment{Construct optimal policy}
        \State $\pi_{\lambda_n - \epsilon} \gets \texttt{SPI}(\lambda_n - \epsilon)$, $\pi_{\lambda_n + \epsilon} \gets \texttt{SPI}(\lambda_n + \epsilon)$
        \State \textsc{Compute} the mixture coefficient $\alpha$ by:
        \State $\beta \gets (F_{\max} - F^{\lambda + \epsilon})/ (F^{\lambda - \epsilon} - F^{\lambda + \epsilon})$
        \State $\alpha \gets \beta \mu_{\pi_{\lambda - \epsilon}}(i_0) / \big(\beta \mu_{\pi_{\lambda - \epsilon}}(i_0) + (1-\beta) \mu_{\pi_{\lambda + \epsilon}}(i_0)\big)$
        \State \Return $\pi^* \gets (\alpha, \pi_{\lambda_n - \epsilon}, \pi_{\lambda_n + \epsilon})$
    \Else \Comment{Shrink search interval}
    \State $\lambda_n^+ \gets \lambda_n$ if $F^{\lambda_n} \leq F_{\max}$; otherwise $\lambda_n^- \gets \lambda_n$
    \EndIf
\EndFor
\end{algorithmic}
\end{algorithm}

\begin{figure*}[t!]
    \centering
    \subfloat[Bisection search]{
        \scalebox{1}{\begin{tikzpicture}[scale=0.95, every node/.style={font=\small}]
  \draw[->, thick] (0,0) -- (7.5,0) node[below] {$\lambda$};
  \draw[->, thick] (0,0) -- (0,4) node[left] {$F^\lambda$};
  \node at (-0.4, 1.8) {$F_{\max}$};
  \draw[thick] (0., 1.8) -- (0.15, 1.8);

  \draw[thick, ->] (2, 0.4) -- (2, 0) node at (2, 0.6) {$\lambda^*$};

  \node at (0,-0.2) {$0$};
  \node at (3.4,-0.2) {$\lambda_1$};
  \draw[thick] 
    ($(3.4,1.2)+(-0.1,-0.1)$) -- ($(3.4,1.2)+(0.1,0.1)$)
    ($(3.4,1.2)+(-0.1,0.1)$) -- ($(3.4,1.2)+(0.1,-0.1)$);
  
  \draw[thick] (3.4,0.) -- (3.4, 0.15);
  \node at (1.7,-0.2) {$\lambda_2$};
  \draw[thick] 
    ($(1.7,2.3)+(-0.1,-0.1)$) -- ($(1.7,2.3)+(0.1,0.1)$)
    ($(1.7,2.3)+(-0.1,0.1)$) -- ($(1.7,2.3)+(0.1,-0.1)$);
    
  \draw[thick] (1.7,0.) -- (1.7, 0.15);
  \node at (2.55,-0.2) {$\lambda_3$};
  \draw[thick] 
    ($(2.55,1.2)+(-0.1,-0.1)$) -- ($(2.55,1.2)+(0.1,0.1)$)
    ($(2.55,1.2)+(-0.1,0.1)$) -- ($(2.55,1.2)+(0.1,-0.1)$);
    
  \draw[thick] (2.55,0.) -- (2.55, 0.15);
  \node at (6.8,-0.2) {$\lambda_{\max}$};
  \draw[thick] (6.8,0.) -- (6.8, 0.15);

  \coordinate (l0) at (0,3.5);
  \coordinate (l0_) at (0.6,3.5);
  \coordinate (l1) at (0.6,2.3);
  \coordinate (l1_) at (2,2.3);
  \coordinate (l2) at (2, 1.2);
  \coordinate (l2_) at (5, 1.2);
  \coordinate (l3) at (5,0.6);
  \coordinate (l3_) at (6.8,0.6);
  \coordinate (lmax) at (6.8,.05);

  \draw[thick] (l0) -- (l0_);
  \draw[thick] (l1) -- (l1_);
  \draw[thick] (l2) -- (l2_);
  \draw[thick] (l3) -- (l3_);
  \draw[thick] (lmax) -- (7.2,.05);

  \draw[dashed] (l0_) -- (l1);
  \draw[dashed] (l1_) -- (l2);
  \draw[dashed] (l2_) -- (l3);
  \draw[dashed] (l3_) -- (lmax);

  \filldraw[black] (l0) circle (2pt);
  \filldraw[black] (l1) circle (2pt);
  \filldraw[black] (l2) circle (2pt);
  \filldraw[black] (l3) circle (2pt);
  \filldraw[black] (lmax) circle (2pt);

  \draw[black, thick] (l0_) circle (2pt);
  \draw[black, thick] (l1_) circle (2pt);
  \draw[black, thick] (l2_) circle (2pt);
  \draw[black, thick] (l3_) circle (2pt);

  \coordinate (g1) at (1.7,3.2);
  \coordinate (g2) at (2.5,3.25);
  \coordinate (g1base) at (1.7,0);
  \coordinate (g2base) at (3.2,0);

\end{tikzpicture}}
        \label{fig:bisection-search}
    }
    \hspace{1cm}
    \subfloat[Intersection search]{
        \scalebox{1}{\begin{tikzpicture}[scale=0.95, every node/.style={font=\small}]
  \draw[->, thick] (0,0) -- (7.5,0) node[below] {$\lambda$};
  \draw[->, thick] (0,0) -- (0,4) node[left] {$\mathcal{L}^\lambda$};
  \draw[thick, ->] (2, 0.4) -- (2, 0) node at (2, 0.6) {$\lambda^*$};

  \node at (0,-0.2) {$0$};
  \node at (1.6,-0.2) {$\lambda_1$};
  \draw[thick] (1.6,0.) -- (1.6, 0.15);
  \node at (2.5,-0.2) {$\lambda_2$};
  \draw[thick] (2.5,0.) -- (2.5, 0.15);
  \node at (2,-0.2) {$\lambda_3$};
  \draw[thick] (2,0.) -- (2, 0.15);
  \node at (6.8,-0.2) {$\lambda_{\max}$};
  \draw[thick] (6.8,0.) -- (6.8, 0.15);

  \coordinate (l0) at (0,0.5);
  \coordinate (l1) at (0.6,1.5);
  \coordinate (l2) at (2, 2.8);
  \coordinate (l3) at (5,3.4);
  \coordinate (lmax) at (6.8,3.5);

  \draw[thick] (l0) -- (l1) -- (l2) -- (l3) -- (lmax);

  \coordinate (g1) at (1.6,3.2);
  \coordinate (g2) at (2.5,3.25);

  \draw[red, thick] (g1) +(-0.08,-0.08) rectangle +(0.1,0.1) node[above left] {$(\lambda_1, \tilde{\mathcal{L}}^{\lambda_1})$};
  \draw[red, thick] (g2) +(-0.08,-0.08) rectangle +(0.1,0.1) node[above right] {$(\lambda_2, \tilde{\mathcal{L}}^{\lambda_2})$};
  \draw[red, thick] (l2) +(-0.096,-0.096) rectangle +(0.11,0.11) node[below right] {$(\lambda_3, \tilde{\mathcal{L}}^{\lambda_3})$};

  \filldraw[black] (l0) circle (2pt);
  \filldraw[black] (l1) circle (2pt);
  \filldraw[black] (l2) circle (1.75pt);
  \filldraw[black] (l3) circle (2pt);
  \filldraw[black] (lmax) circle (2pt);

  \draw[red, dashed] (l1) -- (g1);
  \draw[red, dashed] (g1) -- (l3);
  \draw[red, dashed] (l2) -- (g2);

  \coordinate (g1base) at (1.7,0);
  \coordinate (g2base) at (3.2,0);

\end{tikzpicture}}
        \label{fig:intersection-search}
    }
    \caption{Illustration of the bisection and intersection search methods. The function $F^\lambda$ is piecewise constant and non-increasing, while $\mathcal{L}^\lambda$ is piecewise linear and concave. Crosses in (a) indicate bisection points, and squares in (b) denote intersection points. The optimal Lagrange multiplier $\lambda^*$ corresponds to a breakpoint of $F^\lambda$ and a corner point of $\mathcal{L}^\lambda$.}
    \label{fig:search-algorithm}
\end{figure*}

\subsection{Structured Policy Iteration} 
Computing a $\lambda$-optimal policy using classical dynamic programming methods can be computationally expensive. Methods such as policy iteration (PI) update the recursions in~\eqref{eq:Q-factor} in an unstructured manner~\cite{puterman1994markov}, requiring the evaluation of all state-action pairs when updating the Q-factor. Luo and Pappas~\cite{luo2025cost} proposed an SPI algorithm to compute threshold-like policies. In our problem, $\pi_\lambda(s)$ is non-decreasing in the error holding time $\delta$. If $\pi_\lambda(s)=1$, then the optimal action for all $s^\prime \succeq s$ is to transmit without further computation. To exploit this known structure, SPI evaluates all possible policies in a structured manner; specifically, it searches for the threshold value for each error $(x, z, \theta)$ in increasing order of $\delta$. 

For numerical tractability, the algorithm operates over a finite state space. We truncate the age and semantic variable with sufficiently large bounds, i.e., 
\begin{equation*}
    \tilde{\Theta}_t = \min\{\Theta_t, \Theta_{\max}\}~\textrm{and}~\tilde{\Delta}_t = \min\{\Delta_t, \Delta_{\max}\}.
\end{equation*}
Let $\tilde{\mathbb{S}}$ denote the state space of $\tilde{S}_t = (X_t, Z_{t-1}, \tilde{\Theta}_{t-1}, \tilde{\Delta}_{t-1})$. The rationale for truncation is as follows: 
\begin{itemize}
    \item \emph{Truncation of AoI:} For finite-state irreducible Markov chains, the belief vector $Q^{\Theta_t}_{Z_t}$ converges exponentially fast in the AoI $\Theta_t$, as illustrated in Figure~\ref{fig:AoCI}.
    \item \emph{Truncation of AoCE:} Once the MAP estimator enters its steady state, the error persistence probability diminishes as the error holding time $\Delta_t$ increases. See~\cite[Theorem~3]{luo2025cost} for a detailed analysis.
\end{itemize}

We now present the details of the PI algorithm, followed by the proposed SPI algorithm. The PI operates as follows:
\begin{itemize}
    \item[1)] \textit{Initialization:} Select an initial switching policy $\pi^0_\lambda$, a reference state $s_\textrm{ref}$, and truncation bounds $\Theta_{\max}, \Delta_{\max}$.
    \item[2)] \textit{Policy evaluation:} Find a scalar $\mathcal{L}^n$ and a vector $V^n$ by solving the Bellman equation
    \begin{equation*}
        \mathcal{L}^n + V^n(s) = \ell(s, \pi^n_\lambda(s)) + \mathbb{E}[V^n(s^\prime)|s, \pi_\lambda^n(s)] 
    \end{equation*}
    for all $s \in \tilde{\mathbb{S}}$ such that $V^n(s_\textrm{ref}) = 0$.
    \item[3)] \textit{Policy improvement:} Find $\pi^{n+1}_\lambda$ that satisfies for \emph{all} $s$
    \begin{equation*}
        \pi^{n+1}_\lambda(s)
        = \argmin_{u}\big[\ell(s, u) + \mathbb{E}[V^n(s^\prime)|s, u]\big].
    \end{equation*}
    \item[4)] \textit{Termination:} If $\pi^{n+1}_\lambda = \pi^n_\lambda$, the algorithm terminates with $\pi_\lambda = \pi^n_\lambda$; otherwise set $n=n+1$ and return to Step~2.
\end{itemize}

As can be seen in Step 3, PI evaluates all state-action pairs in an unstructured manner. This is computationally inefficient, as the optimal action for all states exceeding the thresholds is simply to transmit without further computation.

Instead, SPI updates the policy in the increasing order of $\delta$. In Step~3, SPI operates as follows:
\begin{itemize}
    \item[3)] \textit{Policy improvement:} For each $(x, z, \theta)$, update $\pi^{n+1}_\lambda(s)$ in the increasing order of $\delta$:
    \begin{itemize}
        \item[a.] Initialize $s = (x, z, \theta, \delta)$ with $\delta = 1$. 
        \item[b.] Update the policy as
        \begin{equation*}
            \pi^{n+1}_\lambda(s)
            = \argmin_{u} \big[\ell(s, a) + \mathbb{E}\big[V^n(s^\prime)|s, u]\big].
        \end{equation*}
        \item[c.] If $\pi^{n+1}_\lambda(s) = 1$, set $\pi^{n+1}_\lambda(s^\prime) = 1$ for all $\delta^\prime \geq \delta$ and proceed to Step~3(a) with an unvisited error. Otherwise, increase $\delta=\delta + 1$ and return to Step~3(b).
    \end{itemize}
\end{itemize}

The SPI is summarized in Algorithm~\ref{alg:SPI}. The complexity of Step~3 is $\mathcal{O}(|\mathbb{X}|^2\Theta_{\max}\Delta_{\max})$, which is significantly smaller than the $\mathcal{O}(|\tilde{\mathbb{S}}|^2|\mathbb{U}|)$ complexity of unstructured PI.

\subsection{Intersection Search}
Finding the optimal multiplier $\lambda^*$ that satisfies Theorem~\ref{theorem:optimal-policy}.ii is computationally prohibitive, as each multiplier update requires solving a high-dimensional MDP. A common approach is the \emph{bisection search}, which exploits the monotonicity of $F^\lambda$ to update the multiplier, as illustrated in Figure~\ref{fig:bisection-search}. The bisection search proceeds as follows:
\begin{enumerate}
    \item[1)] \textit{Initialization:} Choose a sufficiently large $\lambda_{\max}$ such that $F^{\lambda_{\max}} <F_{\max}$. Set interval $I_0 = [\lambda_0^{-},\lambda_0^{+}] = [0, \lambda_{\max}]$.
    \item[2)] \textit{Multiplier update:} Update the multiplier as the middle point of the search interval, i.e.,
    \begin{equation}
        \lambda_{n} = \texttt{Bisec}(\lambda_{n-1}) = \frac{\lambda_{n-1}^- + \lambda_{n-1}^+}{2}.
    \end{equation}
    Solve $\pi_{\lambda_n} = \texttt{SPI}(\lambda_n)$ at multiplier $\lambda_n$ and obtain $F^{\lambda_n}$. Update the search interval as: $I_n = [\lambda_n, \lambda_{n-1}
    ^{+}]$ if $F^{\lambda^n} \geq F_{\max}$; otherwise, $I_n = [\lambda_{n-1}^{-}, \lambda_n]$.
    \item[3)] \textit{Stop criterion:} The algorithm terminates when a desired accuracy is achieved, i.e., $|\lambda_n^{+} - \lambda_n^{-}| < \epsilon$. Otherwise, set $n=n+1$ and go to Step~2.
\end{enumerate}

The bisection search has time complexity $\mathcal{O}(\log_2(\frac{\lambda_{\max}}{\epsilon}))$. It requires many iterations to achieve the desired accuracy.

\begin{figure*}[t!]
    \centering
    \begin{minipage}{0.49\textwidth}
        \centering
        \includegraphics[width=0.95\linewidth]{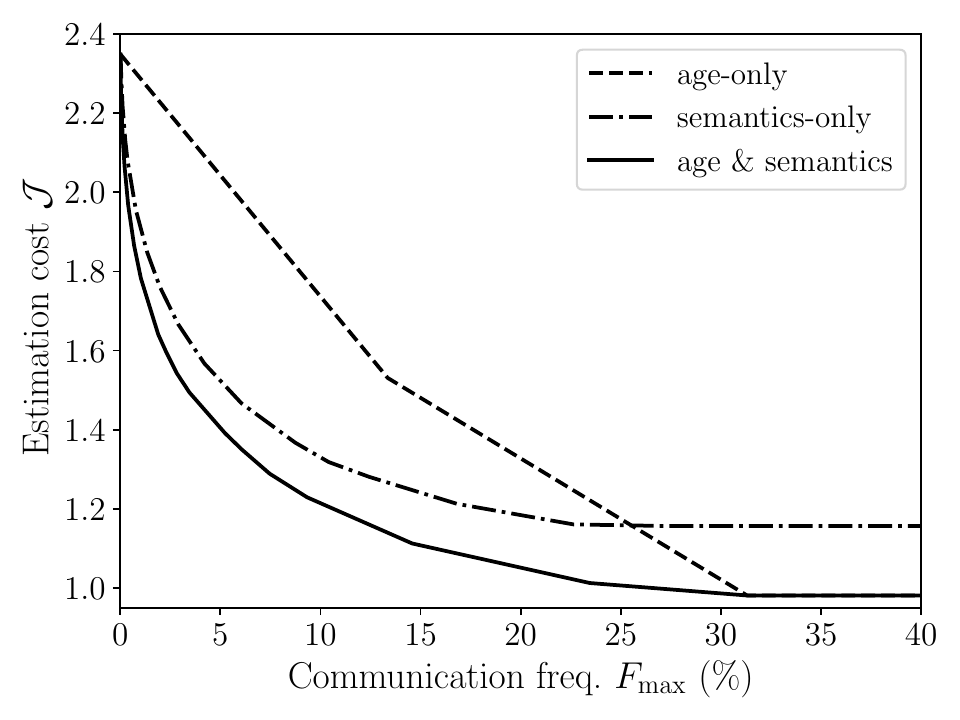}
        \caption{Estimation performance by leveraging different information attributes.}
        \label{fig:avg-cost}
    \end{minipage}
    \hfill
    \begin{minipage}{0.49\textwidth}
        \centering
        \scalebox{0.815}{\begin{tikzpicture}

\pgfmathsetmacro{\lam}{100000}

\begin{axis}[
    width=10.5cm,
    height=8cm,
    xlabel={\large $\textrm{Search tolerance}$ \Large $\epsilon$},
    ylabel={\large $\textrm{Number of iterations}$},
    xmode=log,
    log basis x=10,
    ymin=0, ymax=40,
    xmin=1e-6, xmax=1e-1,
    grid=both,
    minor tick num=1,
    major grid style={line width=.1pt,draw=gray!40},
    minor grid style={line width=.01pt,draw=gray!10},
    xtick={1e-1, 1e-2, 1e-3, 1e-4, 1e-5, 1e-6},
    xticklabels={$10^{-1}$, $10^{-2}$, $10^{-3}$, $10^{-4}$, $10^{-5}$, $10^{-6}$},
    ytick={0,5,10,15,20,25,30,35,40},
    legend style={at={(0.68,0.97)}, anchor=north west},
]
    
\addplot [
    domain=1e-6:1e-1, 
    samples=100, 
    color=black, 
    thick]{log2(\lam/(0.5*x))}; 
    \addlegendentry{\normalsize $\log_2(\lambda/\epsilon)$}
\addplot[
    only marks, 
    mark=x,
    line width=0.8,
    color=black,
    mark size=4.5pt,
]
coordinates {
    (1e-1,21)
    (1e-2,25)
    (1e-3,28)
    (1e-4,31)
    (1e-5,35)
    (1e-6,38)
};\addlegendentry{\normalsize $\textrm{bisection}$}

\addplot[
    only marks,
    mark=square,
    line width=0.8,
    color=black,
    mark size=3pt,
]
coordinates {
    (1e-1,5)
    (1e-2,5)
    (1e-3,5)
    (1e-4,5)
    (1e-5,5)
    (1e-6,5)
};\addlegendentry{\normalsize $\textrm{intersection}$}

\addplot [
    dashed,
    domain=1e-6:1e-1, 
    samples=100, 
    color=black, 
    thick]{5}; 

\end{axis}
\end{tikzpicture}}
        \caption{Time complexity of different multiplier update methods.}
        \label{fig:complexity}
    \end{minipage}
\end{figure*}

Luo and Pappas~\cite{luo2025semantic} proposed an efficient multiplier update method, termed \emph{intersection search}, which significantly outperforms bisection search in terms of time complexity. This method exploits the following properties of $\mathcal{L}^\lambda$: (i) $\mathcal{L}^\lambda$ is piecewise linear and concave, and (ii) the optimal multiplier $\lambda^*$ corresponds to a corner point of $\mathcal{L}^\lambda$, as illustrated in Fig.~\ref{fig:intersection-search}.

The intersection search proceeds as follows:
\begin{enumerate}
    \item[1)] \textit{Initialization:} Choose a sufficiently large $\lambda_{\max}$ such that $F^{\lambda_{\max}} <F_{\max}$. Set interval $I_0 = [\lambda_0^{-},\lambda_0^{+}] = [0, \lambda_{\max}]$.
    \item[2)] \textit{Multiplier update:} Compute the intersection point of two tangents: one is formed by point $(\lambda_{n-1}^{-}, \mathcal{L}^{\lambda_{n-1}^{-}})$ with a slop of $F^{\lambda_{n-1}^{-}}$, and another is formed by point $(\lambda_{n-1}^{+}, \mathcal{L}^{\lambda_{n-1}^{+}})$ with a slop of $F^{\lambda_{n-1}^{+}}$. The intersection point is given by $(\lambda_n, \tilde{\mathcal{L}}^{\lambda_n})$, where
    \begin{subequations}
        \begin{align}
        \lambda_n &= \texttt{Insec}(\lambda_{n-1}) = \frac{\mathcal{J}^{\lambda^{+}_{n-1}} - \mathcal{J}^{\lambda^{-}_{n-1}}}{F^{\lambda^{-}_{n-1}} - F^{\lambda^{+}_{n-1}}}, \label{eq:inter-1}\\
        \tilde{\mathcal{L}}^{\lambda_n} &= \mathcal{J}^{\lambda^{-}_{n-1}} + \lambda_n F^{\lambda^{-}_{n-1}}.\label{eq:inter-2}
    \end{align}
    \end{subequations}
    Solve $\pi_{\lambda_n} = \texttt{SPI}(\lambda_n)$ at multiplier $\lambda_n$ and obtain $F^{\lambda_n}$. Update the search interval as: $I_n = [\lambda_n, \lambda_{n-1}
    ^{+}]$ if $F^{\lambda^n} \geq F_{\max}$; otherwise, $I_n = [\lambda_{n-1}^{-}, \lambda_n]$.
    \item[3)] \textit{Stop criterion:} The algorithm terminates when the intersection point is located on $\mathcal{L}^\lambda$, i.e., $\tilde{\mathcal{L}}^{\lambda_n}=\mathcal{L}^{\lambda_n}$. Otherwise, set $n=n+1$ and go to Step~2.
\end{enumerate}

\begin{figure*}[ht!]
    \centering
    \begin{minipage}{0.49\textwidth}
        \centering
        \includegraphics[width=0.95\linewidth]{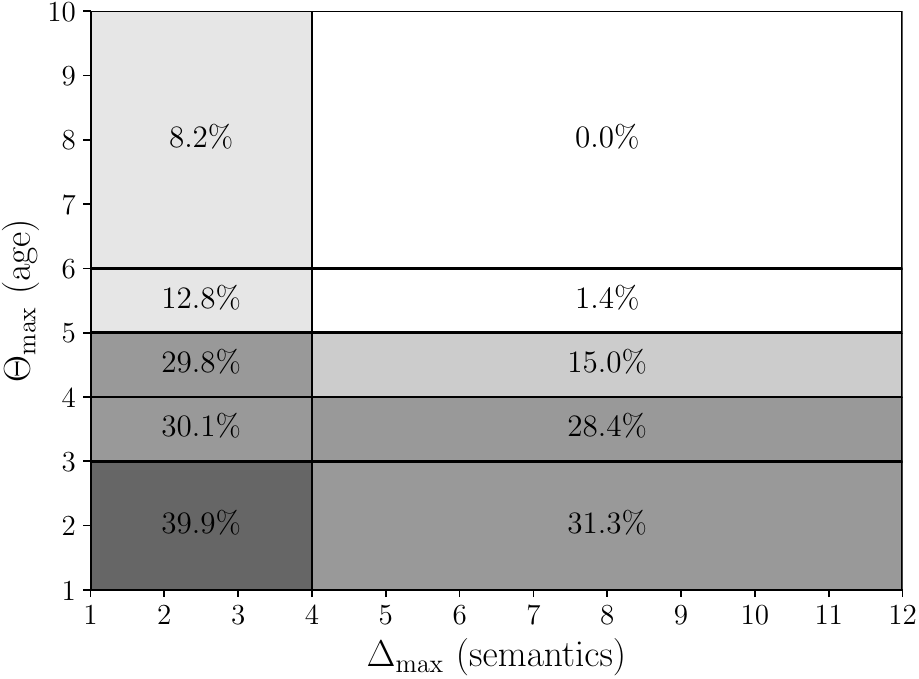}
        \caption{Relative performance gap induced by state space truncation.}
        \label{fig:distance}
    \end{minipage}
    \hfill
    \begin{minipage}{0.49\textwidth}
        \centering
    \includegraphics[width=0.95\linewidth]{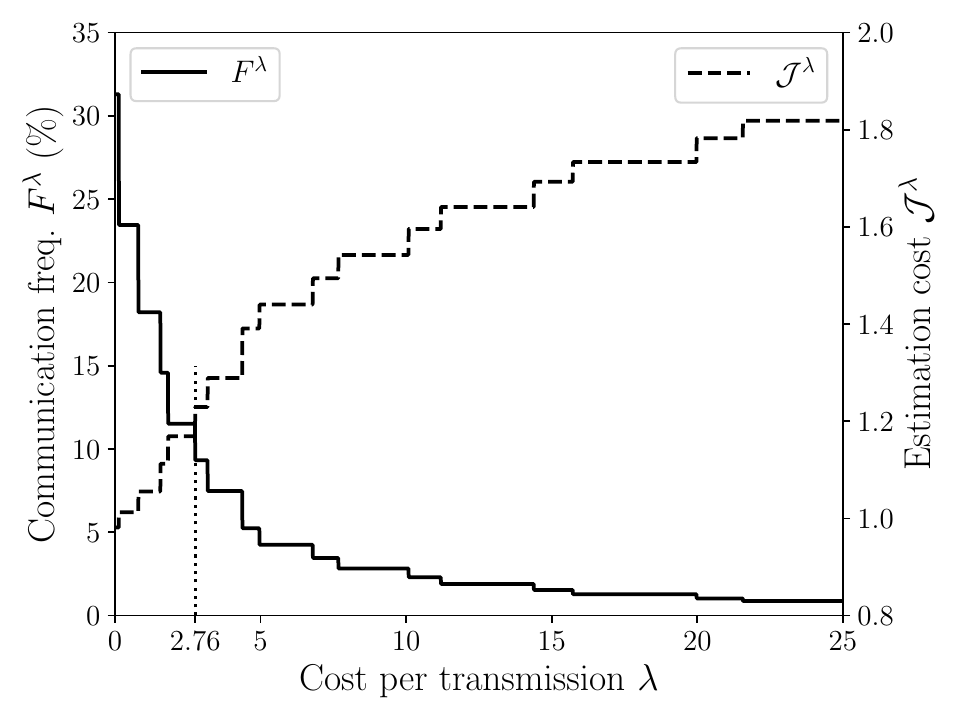}
        \caption{Average costs as a function of multiplier $\lambda$.}
        \label{fig:policy-search}
    \end{minipage}
\end{figure*}

\begin{figure*}[ht!]
    \centering
    \begin{minipage}{0.49\textwidth}
        \centering
        \includegraphics[width=0.95\linewidth]{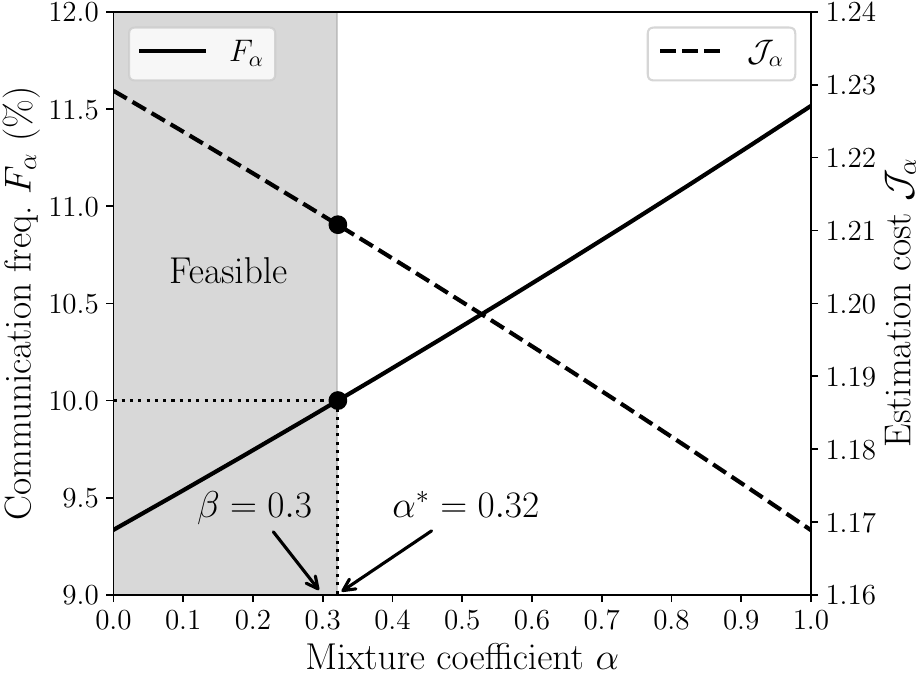}
        \caption{Effect of the mixture coefficient $\alpha$ on the optimal policy.}
        \label{fig:mixure_policy}
    \end{minipage}
    \hfill
    \begin{minipage}{0.49\textwidth}
        \centering
    \includegraphics[width=0.93\linewidth]{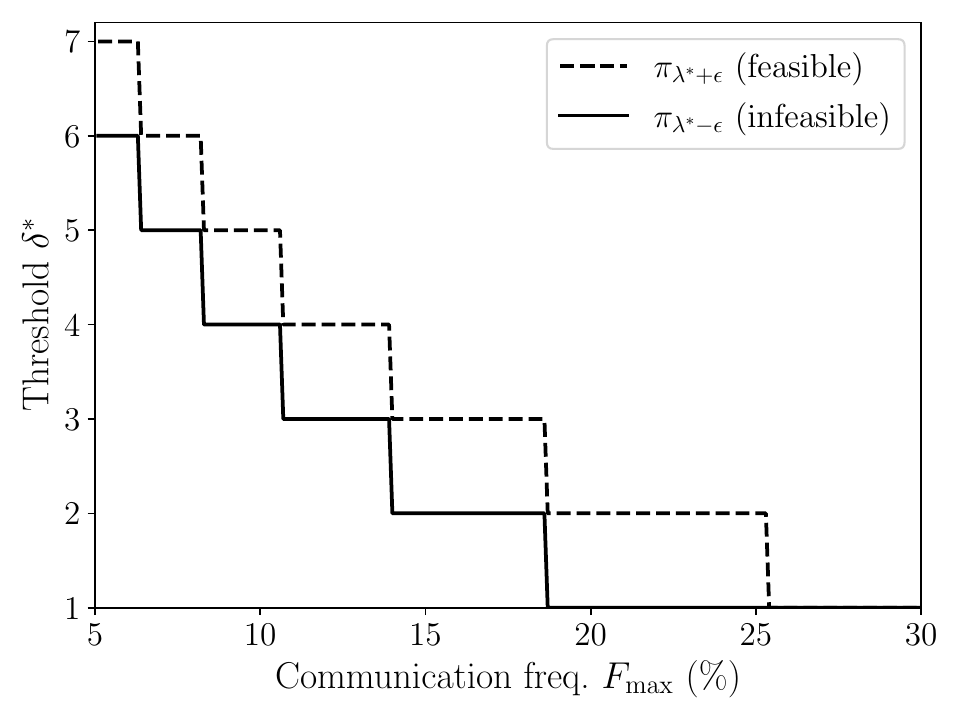}
    \caption{Policy structure of a non-prioritized system.}
    \label{fig:thresholds}
    \end{minipage}
\end{figure*}

Notably, the intersection search method finds the optimal multiplier $\lambda^*$ in only a few iterations because (i) the number of iterations is bounded by the number of segments in $[0, \lambda_{\max}]$, (ii) the time complexity is independent of the accuracy tolerance $\epsilon$, and (iii) choosing a large $\lambda_{\max}$ does not significantly increase the complexity.

\section{Numerical Results}\label{sec:simulations}
This section presents numerical results to validate our theoretical findings on the optimal policy. We consider a source with $|\mathbb{X}| = 3$ states and
\begin{equation*}
    Q= \begin{bmatrix}
        0.85 & 0.1 & 0.05 \\
        0.25 & 0.65 & 0.1 \\
        0.2 & 0.1 & 0.7
    \end{bmatrix}.
\end{equation*}
The urgency of lasting impact is modeled as
\begin{equation*}
    \rho_{i,j}(\delta) = \begin{cases}
        4e^{0.4\delta}, & \textrm{missed alarms}~(i = 1, j\neq 1),\\
        2e^{0.25\delta}, & \textrm{false alarms}~(i \neq 1, j =1),\\
        e^{0.1\delta}, & \textrm{normal errors}~(i\neq j, i,j\neq 1).
    \end{cases}
\end{equation*}
The other parameters are set as follows: the packet-drop probability $p_f = 0.3$, the truncation bounds $\Theta_{\max}=40$ and $\Delta_{\max} = 40$, and the maximum multiplier $\lambda_{\max} = 10^5$. 

\subsection{Performance Comparison}
We compare the proposed solution, which leverages both age and semantics, with two benchmark approaches.
\begin{itemize}
    \item \textit{Age-only strategy:} The receiver uses the MAP estimator, and the sensor minimizes the average distortion. This approach is widely adopted in the control literature.
    \item \textit{Semantics-only strategy:} The receiver uses the ZOH estimator, and the sensor minimizes the average semantics-aware estimation cost. This approach is common in the semantic communication literature.
\end{itemize}

Figure~\ref{fig:avg-cost} compares the estimation performance achieved by leveraging different information attributes. It can be observed that when communication is scarce ($F_{\max} < 0.2$), leveraging the semantic attribute significantly improves estimation performance. This is because the receiver cannot reliably predict the state with sparsely observed samples; therefore, the sensor must prioritize reducing prolonged errors. Conversely, when communication is adequate ($F_{\max} > 0.3$), leveraging information age alone can achieve optimal performance, since the receiver can make reliable predictions with frequent transmissions. Moreover, in this regime, the optimal policy yields the same performance as the always-transmit policy. This suggests that increasing the transmission frequency does not necessarily improve estimation performance, as many updates are of marginal importance to the overall system performance.

Notably, the proposed strategy achieves the desired estimation performance with substantially reduced communication overhead. For instance, an average cost of $\mathcal{J} = 1.2$ can be achieved using approximately 50\% fewer transmissions than the benchmark approaches.

\subsection{Computational Complexity}
Figure~\ref{fig:complexity} compares the complexity of the intersection search and bisection search methods. The intersection search converges once the intersection point lies exactly on $\mathcal{L}^\lambda$; hence, its complexity is independent of the search tolerance $\epsilon$. In contrast, the number of multiplier updates required by the bisection search increases logarithmically as $\epsilon \to 0$. This leads to a substantial computational burden since each update involves solving a high-dimensional MDP.

\subsection{Asymptotic Optimality} 
Figure~\ref{fig:distance} illustrates the relative performance gap introduced by truncation. We observe that the performance converges for moderate truncation bounds. Importantly, increasing the age truncation bound $\Theta_{\max}$ offers two main benefits. First, a larger $\Theta_{\max}$ significantly reduces the performance gap, as the receiver can exploit a longer history to improve predictability. Second, a larger $\Theta_{\max}$ helps prevent prolonged consecutive errors, so that a relatively small $\Delta_{\max}$ suffices to achieve a close approximation. In contrast, when $\Theta_{\max}$ is small, the resulting policy deviates significantly from optimality. In this regime, the receiver cannot make reliable predictions; consequently, given the constraint on communication frequency, the transmitter must prioritize correcting prolonged errors. As a result, many transient errors remain unaddressed, leading to degraded estimation performance.

\subsection{Policy Structure} 
Figure~\ref{fig:policy-search} illustrates the average costs as a function of multiplier $\lambda$. It is observed that $F^\lambda$ and $\mathcal{J}^\lambda$ are monotonic and piecewise constant functions. By Theorem~\ref{theorem:optimal-policy}.i, a deterministic optimal policy exists only when $F_{\max}$ is chosen from the value domain of $F^\lambda$. For instance, when $F_{\max} = 11.52\%$, we obtain $F^{\lambda^*} = F_{\max}$ for all $\lambda^* \in [1.83, 2.76)$. In this case, the corresponding switching policy is optimal. Otherwise, the optimal policy is a simple mixture policy. For example, when $F_{\max} = 10 \% $, there exists no $\lambda^*>0$ such that $F^{\lambda^*} = F_{\max}$. Instead, the optimal multiplier is found at $\lambda^* = 2.76$ which satisfies $F^{\lambda^*+\epsilon} < F_{\max} < F^{\lambda^*-\epsilon}$ for arbitrarily small $\epsilon>0$. From Theorem~\ref{theorem:optimal-policy}.ii, the optimal policy is a simple mixture policy $\pi^* = (\alpha^*, \pi_{\lambda^*-\epsilon}, \pi_{\lambda^*+\epsilon})$, where the optimal mixture coefficient is obtained by $\alpha^* = 0.32$, and $\pi_{\lambda^*-\epsilon}$ (infeasible) and $\pi_{\lambda^*+\epsilon}$ (feasible) are neighboring switching policies that differ in exactly one threshold. Let
\[
\mathcal{J}_\alpha = \mathcal{J}(\alpha, \pi_{\lambda^*-\epsilon}, \pi_{\lambda^*+\epsilon}), \,\,F_\alpha = F(\alpha, \pi_{\lambda^*-\epsilon}, \pi_{\lambda^*+\epsilon}).
\]

Figure~\ref{fig:mixure_policy} shows the impact of the mixture coefficient $\alpha$ on the policy performance. As $\alpha$ increases, the resulting mixture policy relies more on the infeasible policy $\pi_{\lambda^*-\epsilon}$, resulting in a lower estimation cost. The optimal mixture coefficient $\alpha^*$ is attained when the frequency constraint is active, i.e., when $F_{\alpha^*} = 10\%$. We also observe that $F_\beta \neq 10\%$ (where $\beta$ is defined in Theorem~\ref{theorem:optimal-policy}.ii), implying that the mixture policy $\tilde{\pi} = (\beta,\pi_{\lambda^*-\epsilon},\pi_{\lambda^* +\epsilon})$ is suboptimal. Table~\ref{table:coefficient} reports additional results for different values of $F_{\max}$. The results show that the optimal mixture policy $\pi^*$ consistently achieves a slightly lower cost than the approximate policy $\tilde{\pi}$.

\begin{table}[t!]
\centering
\caption{Comparison of Different Mixture Policies}
\label{table:coefficient}
\begin{tabular}{c |ccccccc}
    \toprule
     $F_{\max}$ (\%) & $5$ & $10$ & $15$ & $20$ & $25$ & $30$ \\
    \midrule
    $\mathcal{J}(\pi^*)$ & 1.4029& 1.2108& 1.1059& 1.0408& 1.0012 & 0.9829 \\
    $\mathcal{J}(\hat{\pi})$ & 1.4032& 1.2118& 1.1066& 1.0418& 1.0023 & 0.9831 \\
    \bottomrule
\end{tabular}
\end{table}

Finally, we showcase the thresholds of the optimal policy by examining a non-prioritized system as described in Corollary~\ref{coro:special-case}. We set $M_{i,j} = 0.1$ and $\rho_{i, j}(\delta) = 1.2 e^{0.55\delta} + 0.3$ for all $i\neq j$. As shown in Figure~\ref{fig:thresholds}, when $F_{\max}< 25\%$, Condition (ii) of Proposition~\ref{proposition:ross} is active, and the optimal policy is a mixture of two threshold policies. The threshold of $\pi_{\lambda^*+\epsilon}$ is larger than that of $\pi_{\lambda^*-\epsilon}$ by exactly one. When $F_{\max}>25\%$, Condition (i) of Proposition~\ref{proposition:ross} holds, and the reactive policy (which transmits whenever an error occurs) is optimal.

\section{Conclusion}\label{sec:conclusion}
In this paper, we leveraged both the age and semantics of information in the estimation process to reduce unnecessary transmissions for predictable or low-impact errors. We showed that the optimal communication policy is a simple mixture of two neighboring switching policies whose thresholds depend on the instantaneous error and AoI, and that the structure simplifies to threshold policies in symmetric and non-prioritized systems. We derived a closed-form expression of the mixture coefficient. Building on these results, we developed the Insec-SPI algorithm, which computes the optimal policy with significantly reduced computational overhead. Numerical results demonstrate that jointly exploiting age and semantic attributes significantly improves estimation performance.

\appendices
\renewcommand{\theequation}{A\arabic{equation}}
\setcounter{equation}{0} 

\section{Proof of Theorem~\ref{theorem:existence}}\label{proof:theorem existence}
\textit{Part (i):} Fix an arbitrary $\lambda \geq 0$. It suffices to show that under Assumption~\ref{assumption:existence} there exists at least one admissible policy that yields a finite average cost. Consider an reactive policy $\tilde{\pi}$ that initiates transmission whenever an error occurs. Fix a realization in which the source is in state $x$ and the receiver holds an incorrect estimate $z$ with error holding time $\theta_0$. Since the Markov chain is finite and irreducible, the MAP estimator enters a steady state after a finite number of time steps. Assume that the estimator operates in its steady state, i.e., $\hat{x} = g(z, \theta) = g(z, \theta_0)$ for all $\theta\geq \theta_0$. 

Conditioned on this steady-state behavior, the error persists if the source remains in state $x$ (with probability $M_{x, x}$), and the transmission fails (with probability $p_f$). Thus, the error persistence probability is $M_{x,x}p_f$, and the probability that the error persists for exactly $n$ consecutive time slots is
\begin{equation*}
    p_{x,\hat{x}}(n) = (1-M_{x,x}p_f)(M_{x,x}p_f)^n,
\end{equation*}
The persistence cost along this sample path is given by
\begin{equation*}
    C_{x,\hat{x}}(n) = 
    \sum_{t=0}^n \rho_{x, \hat{x}}(\delta_0+t). 
\end{equation*}
Taking expectation yields
\begin{align*}
\mathcal{J}(\tilde{\pi}) &= \sum_{x\neq\hat{x}}
    \sum_{n=0}^{\infty}p_{x,\hat{x}}(n)C_{x,\hat{x}}(n) \\
    &= \sum_{x\neq \hat{x}}(1-M_{x,x}p_f)\sum_{n=0}^{\infty}\xi_{x, \hat{x}}(n),
\end{align*}
where $\xi_{x, \hat{x}}(n) = (M_{x,x}p_f)^n\rho_{x, \hat{x}}(\delta_0+n)$. The proof is therefore reduced to showing that the summation $\sum_{n=0}^{\infty}\xi_{x, \hat{x}}(n)$ is finite for all $x\neq \hat{x}$. By the ratio test, this holds if
\begin{equation*}
    \lim_{n\to \infty}\frac{\xi_{x, \hat{x}}(n+1)}{\xi_{x, \hat{x}}(n)} = M_{x,x}p_f \lim_{n^\prime\to\infty}\frac{\rho_{x, \hat{x}}(n^\prime + 1)}{\rho_{x, \hat{x}}(n^\prime)}< 1,
\end{equation*}
where $n^\prime = n + \delta_0$. Rearranging the terms gives
\begin{equation*}
    \lim_{n^\prime \to \infty} \frac{\rho_{x, \hat{x}}(n^\prime + 1)}{\rho_{x, \hat{x}}(n^\prime)} < \frac{1}{M_{x, x}p_f}, \quad x \neq \hat{x}.
\end{equation*}
Hence, $\mathcal{J}(\tilde{\pi}) < \infty$. By the optimality of $\pi_\lambda$, we have 
\begin{equation*}
    \mathcal{L}^\lambda \leq \mathcal{J}(\tilde{\pi}) + \lambda F(\tilde{\pi}) < \infty,
\end{equation*}
which establishes assertion (i).

\textit{Part~(ii):} Suppose, for contradiction, that $F^\lambda>F_{\max}$ for all $\lambda\geq 0$. Then, for every $\lambda$, 
\begin{equation}
    \mathcal{L}^\lambda = \mathcal{J}^\lambda + \lambda F^\lambda > \mathcal{J}^\lambda + \lambda F_{\max}. \label{eq:contradiction-1}
\end{equation}
On the other hand, we can construct a feasible policy $\hat{\pi}$ such that 
\begin{equation*}
    F(\hat{\pi}) = F_{\max} - \hat{\epsilon}, \quad \hat{\epsilon} \in (0, F_{\max}),
\end{equation*}
and the Lagrangian cost under $\hat{\pi}$ is
\begin{equation}
    \mathcal{L}^\lambda(\hat{\pi}) = \mathcal{J}(\hat{\pi}) + \lambda (F_{\max} - \hat{\epsilon}). \label{eq:contradiction-2}
\end{equation}
Comparing~\eqref{eq:contradiction-1} and~\eqref{eq:contradiction-2} gives that $\mathcal{L}^\lambda> \mathcal{L}^\lambda(\hat{\pi})$ for sufficiently large $\lambda$, which contradicts the optimality of $\pi_\lambda$. Hence, there exists some $\lambda\geq 0$ such that $F^\lambda\leq F_{\max}$. 

\renewcommand{\theequation}{B\arabic{equation}}
\setcounter{equation}{0} 

\section{Proof of Lemma~\ref{lemma:information-state}}\label{proof:lemma-information-state}
We show that $\{S_t\}$ is a controlled Markov process with control $\{U_t\}$~\cite{mahajan2016decentralized}. That is, the cost and the transition probability depend only on the current information state $S_t$ and the control action $U_t$, and not on the entire history $I_t$, i.e.,
\begin{align}
    \Pr[S_{t+1}|I_t, U_t] &= \Pr[S_{t+1}|S_t, U_t]\label{eq:markov-property},\\
    \mathbb{E}[c(X_t, \hat{X}_t, \Delta_t)|I_t, U_t] &= \mathbb{E}[c(X_t, \hat{X}_t, \Delta_t)|S_t, U_t].\label{eq:absorbing-property}
\end{align}
Since
\begin{align*}
     &\Pr[S_{t+1}|I_t, U_t] \notag\\
     =& \Pr[X_{t+1}, Z_t, \Theta_t, \Delta_t|X_{1:t}, Z_{1:t-1},\Theta_{1:t-1}, \Delta_{1:t-1}, U_{1:t}] \\
     =& \Pr[X_{t+1}|X_t] \times \Pr[Z_t, \Theta_t|X_t, Z_{t-1}, \Theta_{t-1}, U_t] \\
     &\times \Pr[\Delta_t|X_t, Z_t, \Theta_t, \Delta_{t-1}, U_t]
\end{align*}
depends on $I_t$ only through $S_t$, \eqref{eq:markov-property} is true. We now verify~\eqref{eq:absorbing-property}. Recall that $\hat{X}_t = g(Z_t, \Theta_t)$. Therefore, the expectation on the left-hand side of \eqref{eq:absorbing-property} satisfies
\begin{align*}
    &\mathbb{E}[c(X_t, \hat{X}_t, \Delta_t)|I_t, U_t]\\
    =& \sum_{z,\theta,\delta} \Pr[Z_t=z, \Theta_t=\theta|X_t, Z_{t-1}, \Theta_{t-1}, U_t] \\
    &\times \Pr[\Delta_t=\delta|X_t, Z_t = z, \Theta_t = \theta, \Delta_{t-1}, U_t]\times \rho_{X_t, \hat{X}_t}(\delta),
\end{align*}
which implies \eqref{eq:absorbing-property}.

\renewcommand{\theequation}{C\arabic{equation}}
\setcounter{equation}{0} 

\section{Proof of Theorem~\ref{theorem:lambda-optimal}}\label{proof:theorem-lambda-optimal}
Establishing the structural results in Theorem~\ref{theorem:lambda-optimal} requires the submodularity property of the Q-factor \cite[Lemma~4.7.1]{puterman1994markov}.

\begin{definition}
The Q-factor $Q(s, u)$ is submodular if
\begin{equation}
    Q(s^\prime, 1) + Q(s, 0) - Q(s^\prime, 0) - Q(s, 1)\leq 0, \,\, s \preceq s^\prime.\label{eq:submodularity}
\end{equation}
\end{definition}

\begin{lemma}\label{lemma:submodularity}
If $Q(s, u)$ is submodular, then the $\lambda$-optimal policy $\pi_\lambda(s) = \argmin_{u}Q(s, u)$ is monotone. 
\end{lemma}

The proof reduces to verifying~\eqref{eq:submodularity}. For any state $s = (x, z, \theta, \delta)$, the next state $s^\prime$ evolves as follows:

\textit{Case~1:} The sensor takes idle action ($u=0$). The source transitions to a new state $x^\prime$ with probability (w.p.) $M_{x,x^\prime}$ and the receiver's information ages by $1$. Then, in the current time slot, the system may (i) become synced, i.e.,
\begin{equation*}
    s^\prime = (x^\prime, z, \theta+1, 0)~\textrm{if}~\hat{x}=g(z,\theta+1)=x;
\end{equation*}
(ii) remain in the same error, i.e.,
\begin{equation*}
    s^\prime = (x^\prime, z, \theta+1,\delta+1)~\textrm{if}~\hat{x}=g(z,\theta+1)=g(z,\theta);
\end{equation*}
or (iii) jumps into a new error, i.e.,
\begin{equation*}
    s^\prime = (x^\prime,z, \theta+1, 1)~\textrm{if}~
    \hat{x} = g(z,\theta+1) \notin \{x, g(z,\theta)\}.
\end{equation*}
The Q-factor under idle action $(u=0)$ is then given by
\begin{align}
    &Q(s,0) = \mathds{1}_{\hat{x} = x} \sum_{x^\prime}M_{x,x^\prime} V(x^\prime, z, \theta+1, 0)\notag\\
    +& \mathds{1}_{\hat{x}=g(z,\theta)}
    \Big(\rho_{x,\hat{x}}(\delta+1) \hspace{-0.2em}+\hspace{-0.2em}
    \sum_{x^\prime}M_{x,x^\prime} V(x^\prime, z, \theta+1, \delta+1)\Big)\notag\\
    +& \mathds{1}_{\hat{x} \notin \{x, g(z,\theta)\}} \Big(
       \rho_{x,\hat{x}}(1) 
    \hspace{-0.2em}+\hspace{-0.2em} \sum_{x^\prime}Q_{x,x^\prime}V(x^\prime, z, \theta+1, 1)
    \Big).
    \label{eq:Q_0}
\end{align}

\textit{Case~2:} The sensor triggers transmission ($u=1$). If the transmission fails (w.p. $p_f$), the system evolves identically to the idle case, i.e., $Q(s,1) = \lambda + Q(s,0)$. If the transmission succeeds (w.p. $p_s$), the system becomes synced in the current time slot, i.e., $s^\prime = (x^\prime, x, 0, 0)$. The Q-factor under transmission is given by
\begin{equation}
    Q(s,1) \hspace{-0.2em}=\hspace{-0.2em} \lambda + p_f Q(s,0) + p_s \sum_{x^\prime}M_{x,x^\prime}V(x^\prime, x, 0, 0).\label{eq:Q_1}
\end{equation}
Substituting~\eqref{eq:Q_0} and~\eqref{eq:Q_1} into~\eqref{eq:submodularity} yields
\begin{align}
    &Q(s^\prime, 1) + Q(s, 0) - Q(s^\prime, 0) - Q(s, 1) \notag\\
    =& -p_s \big( Q(s^\prime, 0) - Q(s, 0) \big)\notag\\
    =& -p_s \mathds{1}_{\hat{x}=g(z,\theta)}
    \Big( \rho_{x,\hat{x}}(\delta^\prime+1) - \rho_{x,\hat{x}}(\delta+1) \notag\\
    &+ \sum\nolimits_{x^\prime}M_{x,x^\prime} V_\textrm{d}(\delta^\prime, \delta) \Big),
\end{align}
where $V_\textrm{d}(\delta^\prime, \delta) = V(x^\prime, z, \theta+1, \delta^\prime+1) - V(x^\prime, z, \theta+1, \delta+1)$. Since $\rho_{x,\hat{x}}$ is non-decreasing, we have $\rho_{x,\hat{x}}(\delta^\prime+1) \geq \rho_{x,\hat{x}}(\delta+1)$. Hence, it remains to show that the value function $V(s)$ is monotone (see Definition~\ref{eq:poset}). We establish this result in the following lemma, which completes the proof.

\begin{lemma}\label{lemma:monotonicity of V}
   $V(s)\leq V(s^\prime)$ for all $s \preceq s^\prime$.
\end{lemma}
\begin{IEEEproof}
    We prove the result by induction. Choose $V^0(s)$ to be non-decreasing in $\delta$, i.e., $V^0(s)\leq V^0(s^\prime)$ for all $s \preceq s^\prime$. Assume that $V^{n}(s) < V^n(s^\prime)$ for all $s \preceq s^\prime$. We will show that the monotonicity holds for $n+1$. From~\eqref{eq:Q-factor}, we have
\begin{equation*}
    V^{n+1}(s) = \min_{u}\big[Q^{n+1}(s,u)\big] - \min_{u}\big[Q^{n+1}(s_\textrm{ref},u)\big]. 
\end{equation*}

We first show that $Q^{n+1}(s, u)$ is monotone for fixed $u$. When $u=0$, the Q-factor at iteration $n+1$ satisfies
\begin{align}
     & Q^{n+1}(s^\prime,0) - Q^{n+1}(s,0) \notag\\
    =~&\rho_{x,\hat{x}}(\delta^\prime+1) - \rho_{x,\hat{x}}(\delta+1) +\sum\nolimits_{x^\prime}M_{x,x^\prime} V^n_\textrm{d}(\delta^\prime, \delta) 
    \geq 0.\label{eq:prove-inequality-1}
\end{align}
The above inequality holds by the induction hypothesis.

When $u=1$, we may write
\begin{align}
      &Q^{n+1}(s^\prime,1)  - Q^{n+1}(s,1) \notag\\
    =~& p_f(Q^{n+1}(s^\prime,0) - Q^{n+1}(s,0)) \geq 0.\label{eq:prove-inequality-2}
\end{align}
Combining~\eqref{eq:prove-inequality-1} and~\eqref{eq:prove-inequality-2} gives that $Q^{n+1}(s, u)$ is monotone for fixed $u$. It follows that
\begin{align*}
    V^{n+1}(s) &= \min_{u}\big[Q^{n+1}(s,u)\big] - \min_{u}\big[Q^{n+1}(s_\textrm{ref},u)\big] \notag\\
    & \leq \min_{u}\big[Q^{n+1}(s^\prime,u)\big] - \min_{u}\big[Q^{n+1}(s_\textrm{ref},u)\big]\notag\\
    &= V^{n+1}(s^\prime),
\end{align*}
which completes the proof.
\end{IEEEproof}

\renewcommand{\theequation}{D\arabic{equation}}
\setcounter{equation}{0} 

\section{Proof of Corollary~\ref{coro:symmetric-source} and Corollary~\ref{coro:special-case}}\label{proof:coro symmetric source}
\textit{Proof of Corollary~\ref{coro:symmetric-source}:} Fix any transition matrix $M$ of the form~\eqref{eq:special-source}. We will show that the MAP estimate coincides with the ZOH estimate, i.e., $\hat{X}_t = g(Z_t, \Theta_t) = Z_t$ for all $\Theta_t \geq 1$. This is equivalent to showing that $M^n_{i,i} \geq M^n_{i,j}$ for all $i\neq j$ and all $n\geq 1$. 

This form of transition matrix can be written as
\begin{equation}
    M = (1-|\mathbb{X}|\sigma)\mathbf{I} + \sigma \mathbf{1}\mathbf{1}^\top,
\end{equation}
where $\mathbf{I}$ is the identity matrix and $\mathbf{1}$ is the all-ones column vector. Using the binomial expansion, we have
\begin{align}
    M^n &= \big[(1-|\mathbb{X}|\sigma)\mathbf{I} + \sigma \mathbf{1}\mathbf{1}^\top\big]^n\notag\\
    &= \sum_{k=0}^n\binom{n}{k}\big((1-|\mathbb{X}|\sigma)\mathbf{I}\big)^{n-k}\big(\sigma \mathbf{1}\mathbf{1}^\top \big)^k.
\end{align}
Noting that 
\begin{equation*}
    \big(\sigma \mathbf{1}\mathbf{1}^\top \big)^k = \begin{cases}
        \mathbf{I}, & k= 0,\\
        \sigma^k |\mathbb{X}|^{k-1} \mathbf{1}\mathbf{1}^\top, &k\geq 1,
    \end{cases}
\end{equation*}
we obtain
\begin{align}
    M^n 
    &= (1-|\mathbb{X}|\sigma)^n \mathbf{I} + \frac{\mathbf{1}\mathbf{1}^\top}{|\mathbb{X}|} \sum_{k=1}^n\binom{n}{k} (1-|\mathbb{X}|\sigma)^{n-k} (|\mathbb{X}|\sigma)^k \notag\\
    &= (1-|\mathbb{X}|\sigma)^n \mathbf{I} + \frac{1 - (1-|\mathbb{X}|\sigma)^n}{|\mathbb{X}|}\mathbf{1}\mathbf{1}^\top. 
\end{align}
Since $\sigma \leq 1/|\mathbb{X}|$, it follows that
\begin{equation*}
    M^n_{i,i} - M^n_{i,j} = (1-|\mathbb{X}|\sigma)^n \geq 0, \quad j\neq i.
\end{equation*}
This establishes Corollary~\ref{coro:symmetric-source}.

\textit{Proof of Corollary~\ref{coro:special-case}:} The cost function becomes $c(S_t) = \rho(\Delta_t)$ and $S_t = \Delta_{t-1}$ satisfies the conditions~\eqref{eq:markov-property},~\eqref{eq:absorbing-property} in Appendix~\ref{proof:lemma-information-state}. By Corollary~\ref{coro:symmetric-source}, the thresholds are independent of both AoI and the estimation error.

\bibliographystyle{IEEEtran}
\bibliography{ref}

\begin{thebibliography}{10}
\providecommand{\url}[1]{#1}
\csname url@samestyle\endcsname
\providecommand{\newblock}{\relax}
\providecommand{\bibinfo}[2]{#2}
\providecommand{\BIBentrySTDinterwordspacing}{\spaceskip=0pt\relax}
\providecommand{\BIBentryALTinterwordstretchfactor}{4}
\providecommand{\BIBentryALTinterwordspacing}{\spaceskip=\fontdimen2\font plus
\BIBentryALTinterwordstretchfactor\fontdimen3\font minus \fontdimen4\font\relax}
\providecommand{\BIBforeignlanguage}[2]{{%
\expandafter\ifx\csname l@#1\endcsname\relax
\typeout{** WARNING: IEEEtran.bst: No hyphenation pattern has been}%
\typeout{** loaded for the language `#1'. Using the pattern for}%
\typeout{** the default language instead.}%
\else
\language=\csname l@#1\endcsname
\fi
#2}}
\providecommand{\BIBdecl}{\relax}
\BIBdecl

\bibitem{Brockett-TAC-1997}
W.~S. Wong and R.~W. Brockett, ``Systems with finite communication bandwidth constraints. {I}. state estimation problems,'' \emph{IEEE Trans. Autom. Control}, vol.~42, no.~9, pp. 1294--1299, Sep. 1997.

\bibitem{Schenato-ProcIEEE-2007}
L.~Schenato, B.~Sinopoli, M.~Franceschetti, K.~Poolla, and S.~S. Sastry, ``Foundations of control and estimation over lossy networks,'' \emph{Proc. IEEE}, vol.~95, no.~1, pp. 163--187, Jan. 2007.

\bibitem{WNCSSurvey-2018}
P.~Park, S.~Coleri~Ergen, C.~Fischione, C.~Lu, and K.~H. Johansson, ``Wireless network design for control systems: A survey,'' \emph{IEEE Commun. Surv. Tutor.}, vol.~20, no.~2, pp. 978--1013, Apr. 2018.

\bibitem{Lipsa-TAC-2011-scalar-source}
G.~M. Lipsa and N.~C. Martins, ``Remote state estimation with communication costs for first-order {LTI} systems,'' \emph{IEEE Trans. Autom. Control}, vol.~56, no.~9, pp. 2013--2025, Sep. 2011.

\bibitem{lingshi2013TAC}
J.~Wu, Q.-S. Jia, K.~H. Johansson, and L.~Shi, ``Event-based sensor data scheduling: Trade-off between communication rate and estimation quality,'' \emph{IEEE Trans. Autom. Control}, vol.~58, no.~4, pp. 1041--1046, Apr. 2013.

\bibitem{Aditya-TAC-2017}
J.~Chakravorty and A.~Mahajan, ``Fundamental limits of remote estimation of autoregressive {Markov} processes under communication constraints,'' \emph{IEEE Trans. Autom. Control}, vol.~62, no.~3, pp. 1109--1124, Mar. 2017.

\bibitem{WuShuang-TCNS-2020}
S.~Wu, X.~Ren, Q.-S. Jia, K.~H. Johansson, and L.~Shi, ``Learning optimal scheduling policy for remote state estimation under uncertain channel condition,'' \emph{IEEE Trans. Control Netw. Syst.}, vol.~7, no.~2, pp. 579--591, Jun. 2020.

\bibitem{nikos2021CM}
M.~Kountouris and N.~Pappas, ``Semantics-empowered communication for networked intelligent systems,'' \emph{IEEE Commun. Mag.}, vol.~59, no.~6, pp. 96--102, Jun. 2021.

\bibitem{luo2025information}
J.~Luo, E.~Delfani, M.~Salimnejad, and N.~Pappas, ``From information freshness to semantics of information and goal-oriented communications,'' \emph{arXiv preprint arXiv:2512.12758}, 2025.

\bibitem{assaad2020TON}
A.~Maatouk, S.~Kriouile, M.~Assaad, and A.~Ephremides, ``The age of incorrect information: A new performance metric for status updates,'' \emph{IEEE/ACM Trans. Netw.}, vol.~28, no.~5, pp. 2215--2228, Oct. 2020.

\bibitem{luo2026exploiting}
J.~Luo and N.~Pappas, ``Exploiting data significance in remote estimation of discrete-state {Markov} sources,'' \emph{IEEE Trans. Commun.}, vol.~74, pp. 4569--4582, Feb. 2026.

\bibitem{luo2025cost}
J.~Luo and N.~Pappas, ``On the cost of consecutive estimation error: Significance-aware non-linear aging,'' \emph{IEEE Trans. Inf. Theory}, vol.~71, no.~10, pp. 7976--7989, Oct. 2025.

\bibitem{jiping2023TITS}
J.~Luo, T.~Zhang, R.~Hao, D.~Li, C.~Chen, Z.~Na, and Q.~Zhang, ``Real-time cooperative vehicle coordination at unsignalized road intersections,'' \emph{IEEE Trans. Intell. Transp. Syst.}, vol.~24, no.~5, pp. 5390--5405, May 2023.

\bibitem{Roy2012INFOCOM}
S.~Kaul, R.~Yates, and M.~Gruteser, ``Real-time status: How often should one update?'' in \emph{Proc. IEEE Conf. Comput. Commun. (INFOCOM)}, Mar. 2012, pp. 2731--2735.

\bibitem{kosta2017age}
A.~Kosta, N.~Pappas, and V.~Angelakis, ``Age of information: A new concept, metric, and tool,'' \emph{Found. Trends Netw.}, vol.~12, no.~3, pp. 162--259, Dec. 2017.

\bibitem{RoyYates2021JSAC}
R.~D. Yates, Y.~Sun, D.~R. Brown, S.~K. Kaul, E.~Modiano, and S.~Ulukus, ``Age of information: An introduction and survey,'' \emph{IEEE J. Sel. Areas Commun.}, vol.~39, no.~5, pp. 1183--1210, May 2021.

\bibitem{Telatar2021ITW}
Y.~İnan, R.~Inovan, and E.~Telatar, ``Optimal policies for age and distortion in a discrete-time model,'' in \emph{Proc. IEEE Inf. Theory Workshop (ITW)}, Oct. 2021, pp. 1--6.

\bibitem{Melih2021TON}
M.~Bastopcu and S.~Ulukus, ``Age of information for updates with distortion: Constant and age-dependent distortion constraints,'' \emph{IEEE/ACM Trans. Netw.}, vol.~29, no.~6, pp. 2425--2438, Dec. 2021.

\bibitem{Jayanth2023WiOpt}
J.~S, N.~Pappas, and R.~V. Bhat, ``Distortion minimization with age of information and cost constraints,'' in \emph{Proc. 21st Int. Symp. Model. Optim. Mobile, Ad Hoc, Wireless Netw. (WiOpt)}, Aug. 2023, pp. 1--8.

\bibitem{vishrant2025TMC}
R.~V. Ramakanth, V.~Tripathi, and E.~Modiano, ``Monitoring correlated sources: {AoI}-based scheduling is nearly optimal,'' \emph{IEEE Trans. Mobile Comput.}, vol.~24, no.~2, pp. 1043--1054, Feb. 2025.

\bibitem{jiping2025wearing}
J.~Luo, G.~Stamatakis, O.~Simeone, and N.~Pappas, ``Remote state estimation over a wearing channel: Information freshness vs. channel aging,'' \emph{arXiv preprint arXiv:2501.17473}, 2025.

\bibitem{SunYin-TIT-2020}
Y.~Sun, Y.~Polyanskiy, and E.~Uysal, ``Sampling of the {Wiener} process for remote estimation over a channel with random delay,'' \emph{IEEE Trans. Inf. Theory}, vol.~66, no.~2, pp. 1118--1135, Feb. 2020.

\bibitem{SunYin-TON-2021}
T.~Z. Ornee and Y.~Sun, ``Sampling and remote estimation for the {Ornstein-Uhlenbeck} process through queues: Age of information and beyond,'' \emph{IEEE/ACM Trans. Netw.}, vol.~29, no.~5, pp. 1962--1975, Oct. 2021.

\bibitem{george2019GCWkshps}
G.~Stamatakis, N.~Pappas, and A.~Traganitis, ``Control of status updates for energy harvesting devices that monitor processes with alarms,'' in \emph{Proc. IEEE Globecom Workshops (GC Wkshps)}, Dec. 2019, pp. 1--6.

\bibitem{luo2025semantic}
J.~Luo and N.~Pappas, ``Semantic-aware remote estimation of multiple {Markov} sources under constraints,'' \emph{IEEE Trans. Commun.}, vol.~73, no.~11, pp. 11\,093--11\,105, Nov. 2025.

\bibitem{ayan2025age}
O.~Ayan, J.~Luo, N.~Pappas, and X.~An, ``Age-aware {CSI} acquisition of a finite-state {Markovian} channel,'' in \emph{Proc. IEEE Int. Symp. Pers. Indoor Mob. Radio Commun. (PIMRC)}, Sep. 2025, pp. 1--6.

\bibitem{cosandal2025joint}
I.~Cosandal, S.~Ulukus, and N.~Akar, ``Joint age-state belief is all you need: Minimizing {AoII} via pull-based remote estimation,'' in \emph{Proc. ICC Workshops}, Jun. 2025, pp. 1098--1103.

\bibitem{ErfanTGCN24}
E.~Delfani, G.~J. Stamatakis, and N.~Pappas, ``State-aware timeliness in energy harvesting {IoT} systems monitoring a {Markovian} source,'' \emph{IEEE Trans. Green Commun. Netw.}, vol.~9, no.~3, pp. 977--990, Sep. 2025.

\bibitem{gongpu2022UoI}
G.~Chen, S.~C. Liew, and Y.~Shao, ``Uncertainty-of-information scheduling: A restless multiarmed bandit framework,'' \emph{IEEE Trans. Inf. Theory}, vol.~68, no.~9, pp. 6151--6173, Sep. 2022.

\bibitem{chen2024index}
G.~Chen and S.~C. Liew, ``An index policy for minimizing the uncertainty-of-information of {Markov} sources,'' \emph{IEEE Trans. Inf. Theory}, vol.~70, no.~1, pp. 698--721, Feb. 2024.

\bibitem{Andrea2023JSAIT}
G.~Cocco, A.~Munari, and G.~Liva, ``Remote monitoring of two-state {Markov} sources via random access channels: An information freshness vs. state estimation entropy perspective,'' \emph{IEEE J. Sel. Areas Inf. Theory}, vol.~4, pp. 651--666, Dec. 2023.

\bibitem{costa2015age}
M.~Costa, S.~Valentin, and A.~Ephremides, ``On the age of channel information for a finite-state {Markov} model,'' in \emph{Proc. IEEE Int. Con. Commun. (ICC)}, Jun. 2015, pp. 4101--4106.

\bibitem{nikos2021ICAS}
N.~Pappas and M.~Kountouris, ``Goal-oriented communication for real-time tracking in autonomous systems,'' in \emph{Proc. IEEE Int. Conf. Autom. Syst. (ICAS)}, Aug. 2021.

\bibitem{mehrdad2024TCOM}
M.~Salimnejad, M.~Kountouris, and N.~Pappas, ``Real-time reconstruction of {Markov} sources and remote actuation over wireless channels,'' \emph{IEEE Trans. Commun.}, vol.~72, no.~5, pp. 2701--2715, May 2024.

\bibitem{Mehrdad-JCN-2023}
M.~Salimnejad, M.~Kountouris, and N.~Pappas, ``State-aware real-time tracking and remote reconstruction of a {Markov} source,'' \emph{J. Commun. Netw.}, vol.~25, no.~5, pp. 657--669, Oct. 2023.

\bibitem{luo2024goal}
J.~Luo and N.~Pappas, ``Goal-oriented estimation of multiple {Markov} sources in resource-constrained systems,'' in \emph{Proc. IEEE Int. Symp. Pers. Indoor Mob. Radio Commun. (PIMRC)}, Sep. 2024, pp. 1--6.

\bibitem{Zakeri-Asilomar-2023}
A.~Zakeri, M.~Moltafet, and M.~Codreanu, ``Semantic-aware real-time tracking of a {Markov} source under sampling and transmission costs,'' in \emph{Proc. 57th Asilomar Conf. Signals, Syst., Comput. (ACSSC)}, Oct. 2023, pp. 694--698.

\bibitem{Niu2020TWC}
X.~Zheng, S.~Zhou, and Z.~Niu, ``Urgency of information for context-aware timely status updates in remote control systems,'' \emph{IEEE Trans. Wireless Commun.}, vol.~19, no.~11, pp. 7237--7250, Nov. 2020.

\bibitem{SunYin2023MILCOM}
T.~Z. Ornee, M.~K.~C. Shisher, C.~Kam, and Y.~Sun, ``Context-aware status updating: Wireless scheduling for maximizing situational awareness in safety-critical systems,'' in \emph{IEEE Mil. Commun. Conf. (MILCOM)}, Nov. 2023, pp. 194--200.

\bibitem{ChenYutao-TON-AoII-2024}
Y.~Chen and A.~Ephremides, ``Minimizing age of incorrect information over a channel with random delay,'' \emph{IEEE/ACM Trans. Netw.}, vol.~32, no.~4, pp. 2752--2764, Aug. 2024.

\bibitem{Tony2025TON}
K.~Bountrogiannis, A.~Ephremides, P.~Tsakalides, and G.~Tzagkarakis, ``Age of incorrect information with hybrid {ARQ} under a resource constraint for {N-Ary} symmetric {Markov} sources,'' \emph{IEEE/ACM Trans. Netw.}, vol.~33, no.~2, pp. 640--653, Apr. 2025.

\bibitem{Nail-ISIT-2024}
I.~Cosandal, N.~Akar, and S.~Ulukus, ``{AoII}-optimum sampling of {CTMC} information sources under sampling rate constraints,'' in \emph{Proc. IEEE Int. Symp. Inf. Theory (ISIT)}, Jul. 2024, pp. 1391--1396.

\bibitem{Andrea2024distributed}
F.~Chiariotti, A.~Munari, L.~Badia, and P.~Popovski, ``Distributed optimization of age of incorrect information with dynamic epistemic logic,'' in \emph{Proc. IEEE Conf. Comput. Commun.}, May 2025, pp. 1--10.

\bibitem{luo2024minimizing}
J.~Luo and N.~Pappas, ``Minimizing the age of missed and false alarms in remote estimation of {Markov} sources,'' in \emph{Proc. 25th ACM MobiHoc}, Oct. 2024, pp. 381--386.

\bibitem{kriouile2025semantics}
S.~Kriouile, M.~Assaad, and T.~Soleymani, ``Semantics of anomalies in networked control,'' in \emph{Proc. IEEE Int. Symp. Inf. Theory (ISIT)}, Jun. 2025, pp. 1--6.

\bibitem{krishnamurthy2016POMDP}
V.~Krishnamurthy, \emph{Partially Observed {Markov} Decision Processes: From Filtering to Controlled Sensing}.\hskip 1em plus 0.5em minus 0.4em\relax Cambridge, UK: Cambridge Univ. Press, 2016.

\bibitem{zhao2008myopic}
Q.~Zhao, B.~Krishnamachari, and K.~Liu, ``On myopic sensing for multi-channel opportunistic access: structure, optimality, and performance,'' \emph{IEEE Trans. Wireless Commun.}, vol.~7, no.~12, pp. 5431--5440, Dec. 2008.

\bibitem{liu2010indexability}
K.~Liu and Q.~Zhao, ``Indexability of restless bandit problems and optimality of whittle index for dynamic multichannel access,'' \emph{IEEE Trans. Inf. Theory}, vol.~56, no.~11, pp. 5547--5567, Nov. 2010.

\bibitem{ouyang2016downlink}
W.~Ouyang, A.~Eryilmaz, and N.~B. Shroff, ``Downlink scheduling over markovian fading channels,'' \emph{IEEE/ACM Trans. Netw.}, vol.~24, no.~3, pp. 1801--1812, Jun. 2016.

\bibitem{ross1985CMDP}
F.~J. Beutler and K.~W. Ross, ``Optimal policies for controlled markov chains with a constraint,'' \emph{J. Math. Anal. Appl.}, vol. 112, no.~1, pp. 236--252, 1985.

\bibitem{puterman1994markov}
M.~L. Puterman, \emph{{Markov} Decision Processes: Discrete Stochastic Dynamic Programming}.\hskip 1em plus 0.5em minus 0.4em\relax Hoboken, NJ, USA: Wiley, 1994.

\bibitem{mahajan2016decentralized}
A.~Mahajan and M.~Mannan, ``Decentralized stochastic control,'' \emph{Ann. Oper. Res.}, vol. 241, no.~1, pp. 109--126, 2016.

\end{thebibliography}

\end{document}